\documentclass[letterpaper,11pt]{paper}

\usepackage[margin=1in]{geometry}

\usepackage[utf8x]{inputenc}
\usepackage[T1]{fontenc}
\usepackage{lmodern}
\usepackage{microtype}
\usepackage[english]{babel}
\usepackage{graphicx}
\graphicspath{{figs/}}

\usepackage[usenames,dvipsnames]{xcolor}
\usepackage{etoolbox}
\usepackage{setspace}
\usepackage{subfig}
\usepackage{enumitem}
\usepackage{tabularx}

\usepackage{amsmath,amssymb,amsfonts}
\usepackage{amsthm}
\theoremstyle{plain}
\newtheorem{theorem}{Theorem}[section]
\newtheorem{lemma}[theorem]{Lemma}

\newtheorem{corollary}[theorem]{Corollary}

\newtheorem{observation}[theorem]{Observation}

\theoremstyle{definition}
\newtheorem{definition}[theorem]{Definition}

\theoremstyle{remark}

\usepackage{xspace}
\newcommand{\abbrev}[2]{\expandafter\newcommand\csname #1\endcsname{#2\xspace}}
\abbrev{Caratheodory}{Carath\'eodory}
\abbrev{Barany}{B\'ar\'any}
\abbrev{Matousek}{Matou{\v{s}}ek}
\abbrev{Lovasz}{Lov{\'a}sz}
\abbrev{Frederic}{Fr\'ed\'eric}

\abbrev{suchthat}{s.t.\@}
\abbrev{wrt}{w.r.t.\@}
\abbrev{ie}{i.e.\@}
\abbrev{etal}{et al.\@}

\newcommand{\lt}{\left}
\newcommand{\rt}{\right}

\DeclareFontFamily{U}{mathx}{\hyphenchar\font45}
\DeclareFontShape{U}{mathx}{m}{n}{
      <5> <6> <7> <8> <9> <10>
      <10.95> <12> <14.4> <17.28> <20.74> <24.88>
      mathx10
      }{}
\DeclareSymbolFont{mathx}{U}{mathx}{m}{n}
\DeclareFontSubstitution{U}{mathx}{m}{n}
\DeclareMathAccent{\widecheck}{0}{mathx}{"71}
\DeclareMathAccent{\wideparen}{0}{mathx}{"75}

\newcommand{\down}[1]{\check{#1}}

\newcommand{\up}[1]{{\hat{#1}}}
\newcommand{\UP}[1]{{\widehat{#1}}}

\newcommand{\set}[1]{\left\{#1\right\}}
\newcommand{\midd}{\,\middle\vert\,}
\newcommand{\eps}{\ensuremath{\varepsilon}}

\DeclareMathOperator{\poly}{poly}
\DeclareMathOperator{\aff}{aff}

\DeclareMathOperator*{\argmax}{arg\,max}
\DeclareMathOperator{\conv}{conv}
\newcommand{\DD}{\Delta}

\DeclareMathOperator{\dir}{dir}

\DeclareMathOperator{\inter}{int}
\DeclareMathOperator{\lspan}{span}

\DeclareMathOperator{\pos}{pos}
\DeclareMathOperator{\rank}{rank}
\DeclareMathOperator{\relint}{relint}
\DeclareMathOperator{\sd}{sd}
\DeclareMathOperator{\sgn}{sgn}
\DeclareMathOperator{\pred}{pred}
\DeclareMathOperator{\suc}{succ}

\newcommand{\paranthesis}[1]{\left(#1\right)}
\newcommand{\convv}[1]{\conv\paranthesis{#1}}
\newcommand{\poss}[1]{\pos\paranthesis{#1}}

\newcommand{\N}{\ensuremath{\mathbb{N}}}
\newcommand{\Q}{\ensuremath{\mathbb{Q}}}
\newcommand{\R}{\ensuremath{\mathbb{R}}}
\newcommand{\Rp}{\ensuremath{\mathbb{R}_+}}
\newcommand{\Z}{\ensuremath{\mathbb{Z}}}

\newcommand{\ve}[1]{{\boldsymbol #1}}
\newcommand{\0}{\ve{0}}
\newcommand{\1}{\ve{1}}
\newcommand{\e}{\ve{e}}
\newcommand{\bb}{\ve{b}}
\newcommand{\cc}{\ve{c}}

\newcommand{\mm}{\ve{\mu}}
\newcommand{\pp}{\ve{p}}
\newcommand{\qq}{\ve{q}}
\newcommand{\rr}{\ve{r}}
\newcommand{\er}{\up{\ve{r}}}
\newcommand{\vv}{\ve{v}}
\newcommand{\ww}{\ve{w}}
\newcommand{\xx}{\ve{x}}

\newcommand{\cclasss}[2]{\abbrev{#1}{\textsf{#2}}}
\newcommand{\cclass}[1]{\cclasss{#1}{#1}}
\cclasss{cP}{\#P}
\cclass{NP}
\cclass{coNP}
\cclass{PSPACE}
\cclass{FNP}
\cclass{TFNP}
\cclass{FP}
\cclass{PLS}
\cclass{PPAD}
\cclass{PPA}
\cclass{PPADS}
\cclass{PPP}
\cclass{CLS}

\cclasss{CP}{\#P}
\cclasss{WOne}{W[1]}

\newcommand{\prob}[1]{\textsc{#1}}
\newcommand{\defprob}[2]{\abbrev{#1}{\prob{#2}}}
\defprob{NCP}{Ncp}
\defprob{ThreeSat}{3Sat}
\defprob{NCPl}{L-Ncp}
\defprob{NCPg}{G-Ncp}
\defprob{MTSATl}{Max-2SAT/Flip}
\defprob{Centerpoint}{Centerpoint}
\defprob{SimCenter}{SimplicialCenter}
\defprob{Tverberg}{Tverberg}
\defprob{CCP}{Colorful\Caratheodory}
\defprob{convCCP}{convex-\CCP}

\newcommand{\mc}[1]{\ensuremath{\mathcal{#1}}}

\DeclareMathOperator{\enOperator}{enc}
\newcommand{\en}[1]{{\enOperator\lt(#1\rt)}}
\DeclareMathOperator{\colOperator}{col}
\newcommand{\col}[1]{{\colOperator\lt(#1\rt)}}
\DeclareMathOperator{\suppOperator}{supp}
\newcommand{\supp}[1]{{\suppOperator\lt(#1\rt)}}
\DeclareMathOperator{\indOperator}{ind}
\newcommand{\ind}[1]{{\indOperator\lt(#1\rt)}}

\newcommand{\pc}[1]{\mathcal{#1}}
\newcommand{\QQ}{\pc{Q}}
\renewcommand{\SS}{\pc{S}}
\newcommand{\FS}{\Sigma}

\newcommand{\FF}{\pc{F}}
\newcommand{\PP}{\pc{P}}
\newcommand{\PC}{\pc{P}^\CC}
\newcommand{\MM}{\pc{M}}
\newcommand{\LR}{L^\Phi}
\newcommand{\wv}{\ve{w}}

\defprob{CC}{CC}
\newcommand{\LC}{L^\CC}

\newcommand{\facet}[1]{\down{#1}}

\newcommand{\apx}{\approx}
\newcommand{\CA}{C^\approx}
\newcommand{\ba}{\bb^\approx}

\DeclareMathOperator{\PolyMath}{\mathbb{P}}
\newcommand{\Poly}[1]{\PolyMath\!\left[#1\right]}

\newcommand{\TwoRowVec}[2]{\begin{pmatrix}#1\\ #2\end{pmatrix}}
\newcommand{\Landau}[2]{\expandafter\newcommand\csname #1\endcsname[1]{#2\left(##1\right)}}
\Landau{Oh}{O}
\Landau{Th}{\Theta}
\Landau{Om}{\Omega}
\Landau{tOh}{\widetilde{O}}

\title{The Rainbow at the End of the Line --- A \PPAD Formulation of
  the Colorful \Caratheodory Theorem with Applications}

\author{%
  Fr\'ed\'eric Meunier\thanks{%
    Universit\'e Paris Est, CERMICS (ENPC),
    \texttt{\{frederic.meunier,pauline.sarrabezolles\}@enpc.fr}.
    }
\and
  Wolfgang Mulzer\thanks{%
    Institut f\"ur Informatik, Freie Universit\"at Berlin,
    \{\texttt{mulzer,yannikstein}\}\texttt{@inf.fu-berlin.de}.  WM was
    supported in part by DFG Grants MU 3501/1 and MU 3501/2. YS was
    supported by the Deutsche Forschungsgemeinschaft within the
    research training group `Methods for Discrete Structures' (GRK
    1408) and by GIF grant 1161.
    }
\and
  Pauline Sarrabezolles\footnotemark[1]
\and
  Yannik Stein\footnotemark[2]
}

\date{}

\begin{document}

\maketitle
\thispagestyle{empty}

\begin{abstract}

Let $C_1,...,C_{d+1}$ be $d+1$ point sets in $\R^d$, each containing
the origin in its convex hull. A subset $C$ of $\bigcup_{i=1}^{d+1}
C_i$ is called a colorful choice (or rainbow) for $C_1, \dots,
C_{d+1}$, if it contains exactly one point from each set $C_i$. The
colorful \Caratheodory theorem states that there always exists a
colorful choice for $C_1,\dots,C_{d+1}$ that has the origin in its
convex hull. This theorem is very general and can be used to
prove several other existence theorems in high-dimensional discrete
geometry, such as the centerpoint theorem or Tverberg's theorem. The
colorful \Caratheodory problem (\CCP) is the computational problem of
finding such a colorful choice. Despite several efforts in the past,
the computational complexity of \CCP in arbitrary dimension is still
open.

We show that \CCP lies in the intersection of the complexity classes
\PPAD and \PLS. This makes it one of the few geometric problems in
\PPAD and \PLS that are not known to be solvable in polynomial time.
Moreover, it implies that the problem of computing centerpoints, 
computing Tverberg
partitions, and computing points with large simplicial depth is
contained in $\PPAD \cap \PLS$. This is the first
nontrivial upper bound on the complexity of these problems.

Finally,
we show that our \PPAD formulation leads to a polynomial-time
algorithm for a special case of \CCP in which we have only two color
classes $C_1$ and $C_2$ in $d$ dimensions, each with the origin in its
convex hull, and we would like to find a set with half the points from
each color class that contains the origin in its convex hull.
\end{abstract}

\section{Introduction}
Let $P\subset\R^d$ be a $d$-dimensional point set. We say $P$
\emph{embraces} a point $\pp \in \R^d$ or $P$ is
\emph{$\pp$-embracing} if $\pp \in \conv(P)$, and we
say $P$ \emph{ray-embraces} $\pp$ if $\pp \in \pos(C)$, where $\pos(P)
= \big\{ \sum_{\pp \in P} \alpha_{\pp} \pp
\mid \alpha_{\pp} \geq 0 \text{ for all $\pp \in P$} \big\}$.
\Caratheodory's
theorem~\cite[Theorem~1.2.3]{Matouvsek2002} states that if $P$
embraces the origin, then there exists a subset $P' \subseteq P$ of
size $d+1$ that also embraces the origin. This was generalized by
\Barany~\cite{Barany1982} to the \emph{colorful} setting: let
$C_1,\dots,C_{d+1} \subset \R^d$ be point sets that each embrace the
origin. We call a set $C = \{\cc_1, \dots, \cc_{d+1}\}$ 
a \emph{colorful choice} (or
\emph{rainbow}) for $C_1, \dots, C_{d+1}$,
if $\cc_i \in C_i$, for $i = 1, \dots, d+1$. The
\emph{colorful} \Caratheodory theorem states that there always exists
a $\0$-embracing colorful choice that contains the origin in its
convex hull. \Barany also gave the following generalization.

\begingroup
\newcommand{\scalefactor}{0.6}
\begin{figure}[htbp]
  \centering
  \subfloat[]{\label{fig:colcara:conv}\includegraphics[scale=\scalefactor]{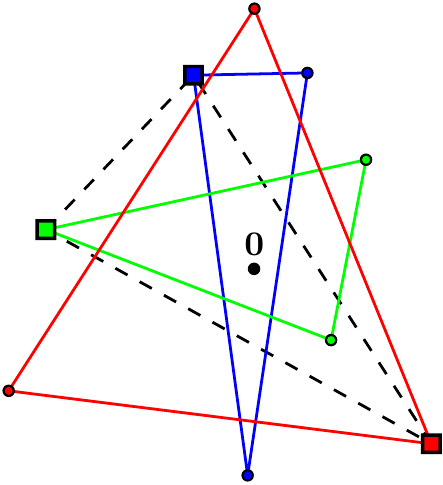}}
  \hspace{2cm}
  \subfloat[]{\label{fig:colcara:cone}\includegraphics[scale=\scalefactor,page=2]{colcara.pdf}}
  \caption{\protect\subref{fig:colcara:conv} Example of the convex version of
    Theorem~\ref{thm:colcara} in two dimensions.
    \protect\subref{fig:colcara:cone} Example of the cone version of
    Theorem~\ref{thm:colcara} in two dimensions.}
  \label{fig:colcara}
\end{figure}
\endgroup

\begin{theorem}[Colorful \Caratheodory Theorem, Cone Version
  \cite{Barany1982}]
\label{thm:colcara}
Let $C_1,\dots,C_d \subset \R^d$ be point sets and $\bb \in \R^d$ a
point with $\bb \in \pos(C_i)$, for $i = 1, \dots, d$. Then, there is a
colorful choice $C$ for $C_1, \dots, C_d$ that ray-embraces $\bb$. \qed
\end{theorem}

The classic (convex) version of the
colorful \Caratheodory theorem follows easily from
Theorem~\ref{thm:colcara}: lift the sets $C_1,\dots,C_{d+1}
\subset \R^d$ to $\R^{d+1}$ by appending a $1$ to each element, and set
$\bb=(0,\dots,0,1)^T$. See Figure~\ref{fig:colcara} for an example of both
versions in two dimensions.

Even though the cone version of the colorful \Caratheodory theorem
guarantees the existence of a colorful choice that ray-embraces the
point $\bb$, it is far from clear how to find it efficiently. We call
this computational problem the \emph{colorful \Caratheodory problem}
(\CCP).  To this day, settling the complexity of \CCP remains an
intriguing open problem, with a potentially wide range of
consequences.  We can use linear programming to check in polynomial
time whether a 
given colorful choice ray-embraces a point, so \CCP
lies in \emph{total function \NP}
(\TFNP)~\cite{Papadimitriou1994}, the complexity class of total 
search problems that
can be solved in non-deterministic polynomial time. This implies that
\CCP cannot be \NP-hard unless
$\NP=\coNP$~\cite{JohnsonPaYa1988}.
However, the complexity landscape
inside \TFNP is far from understood, and there exists a rich body of
work that studies subclasses of \TFNP meant to capture
different aspects of mathematical existence proofs,
such as the pigeonhole principle (\PPP), potential function
arguments (\PLS,
\CLS), or various parity arguments (\PPAD, \PPA,
\PPADS)~\cite{DaskalakisPa11,JohnsonPaYa1988,Papadimitriou1994}.

While the
complexity of \CCP remains elusive, related problems are known
to be complete for \PPAD or for
\PLS. For example, given $d+1$ point sets
$C_1,\dots,C_{d+1} \subset \Q^d$ consisting of two points each and a
colorful choice $C$ for $C_1, \dots, C_{d+1}$ that embraces the
origin, it is \PPAD-complete to find another
colorful choice that embraces the origin~\cite{MeunierSa2014}.
Furthermore, given $d+1$ point sets $C_1,\dots,C_{d+1} \subset \Q^d$,
we call a colorful
choice $C$ for $C_1, \dots, C_{d+1}$
\emph{locally optimal} if the $L_1$-distance of $\conv(C)$ to the
origin cannot be decreased by swapping a point of color $i$ in $C$
with another
point from the same color. Then, computing a locally optimal
colorful choice is \PLS-complete~\cite{MulzerSt2015}.

Understanding the complexity of \CCP becomes even more interesting in the
light of the fact that the colorful \Caratheodory theorem plays a
crucial role in proving several other prominent theorems in convex
geometry, such as Tverberg's theorem~\cite{Sarkaria1992} (and hence
the centerpoint theorem~\cite{Rado1946}) and the first
selection lemma~\cite{Matouvsek2002,Barany1982}. In fact, these proofs
can be
interpreted as polynomial time reductions from the respective
computational
problems, \Tverberg, \Centerpoint, and \SimCenter,
to \CCP. See Section~\ref{sec:app:reductions} for more details.

Several approximation algorithms have been proposed for \CCP. 
\Barany and Onn~\cite{BaranyOn1997} describe an
exact algorithm that can be stopped early to find a colorful choice
whose convex hull is ``close'' to the origin.  More precisely, let 
$\eps, \rho > 0$ be
parameters.  We call a set \emph{$\eps$-close} if its convex hull has
$L_2$-distance at most $\eps$ to the origin.  Given sets
$C_1,\dots,C_{d+1}\subset \R^d$ such that (i) each $C_i$ contains a ball
of radius $\rho$ centered at the origin in its convex hull; and (ii) all
points $\pp \in \bigcup_{i = 1}^{d+1} C_i$ fulfill $1\leq \|\pp\| \leq 2$
and can be encoded using $L$ bits, one can find an
$\eps$-close colorful choice in time
$O(\text{poly}(L, \log(1/\eps),1/\rho))$ on the \textsc{Word-Ram} with
logarithmic costs. For $\eps=0$, the algorithm actually finds a
solution to \CCP in finite time, and, more interestingly, if $1/\rho =
O(\text{poly}(L))$, the algorithm finds a solution to \CCP in
polynomial time.  In the same spirit, Barman~\cite{barman2015} showed
that if the points have constant norm, an $\eps$-close colorful choice
can be found by solving $d^{O(1/\eps^2)}$ convex programs.
Mulzer and Stein~\cite{MulzerSt2015}
considered a different notion of approximation: a set is called
\emph{$m$-colorful} if it contains at most $m$ points from each $C_i$.
They showed that for all fixed $\eps>0$, an $\lceil \eps
d\rceil$-colorful choice that contains the origin in its convex hull
can be found in polynomial time.

\paragraph{Our Results.}
We provide a new upper bound on the complexity of \CCP by showing that
the problem is contained in  $\PPAD \cap \PLS$, implying the first
nontrivial upper bound on the computational complexity of computing
centerpoints or finding Tverberg partitions.

The traditional proofs of the colorful \Caratheodory theorem all
proceed through a potential function argument. Thus, it may not be
surprising that \CCP lies in \PLS, even though a detailed proof 
that can deal with degenerate instances requires
some care (see Section~\ref{sec:app:pls}). 
On the other hand, showing that \CCP lies in \PPAD calls for
a completely new approach. Even though there are proofs of the
colorful \Caratheodory theorem that use topological methods usually
associated with \PPAD (such as certain variants of Sperner's 
lemma)~\cite{holmsen2013,kalai2005}, these proofs 
involve existential
arguments that have no clear algorithmic interpretation. 
Thus,  we present a new proof of the
colorful \Caratheodory theorem that proceeds similarly as the usual
proof for Sperner's lemma~\cite{Cohen1967}.  This new proof has
an algorithmic interpretation that leads to a
formulation of \CCP as a \PPAD-problem.

Finally, we consider the special case of \CCP that we are given 
two color classes 
$C_1, C_2 \subset \R^d$ of $d$ points each and a vector $\bb \in \R^d$
such that both $C_1$ and $C_2$ ray-embrace $\bb$.
We describe an algorithm that solves the following problem in
polynomial time: given $k \in [d]$, find a set $C \subseteq C_1 \cup
C_2$ with $|C \cap C_1| = k$ and $|C \cap C_2| = d-k$ such that
$C$ ray-embraces $\bb$. Note that this is a special
case of \CCP since we can just take $k$ copies of $C_1$ and
$d-k$ copies of $C_2$ in a problem instance for \CCP.

\section{Preliminaries}
\label{sec:prelim}

\paragraph{The Complexity Class \PPAD.}
The complexity class \emph{polynomial parity argument in a directed
graph} (\PPAD)~\cite{Papadimitriou1994} is a subclass of \TFNP that contains
search problems that can be modeled as follows: let $G=(V,E)$ be
a directed graph in which each node has indegree and outdegree at most one. That
is, $G$ consists of paths and cycles.
We call a node $v \in V$ a \emph{source} if $v$ has indegree $0$ and we call $v$ a \emph{sink} if it has outdegree
$0$. Given a source in $G$, we want to find another source
or sink. By a parity argument, there is an even number of sources and
sinks in $G$ and hence another source or sink must exist.
However, finding this sink or source is nontrivial since $G$ is defined
implicitly and the total number of nodes may be exponential.

More formally, a problem in \PPAD is a relation $\mc{R}$ between
a set $\mc{I} \subseteq \{0,1\}^\star$ of \emph{problem instances}
and a set $\mc{S} \subset \{0,1\}^\star$ of \emph{candidate solutions}. Assume
further the following.
\begin{itemize}
\item The set $\mc{I}$ is polynomial-time verifiable.
Furthermore, there is an algorithm that on input $I \in \mc{I}$ and $s
\in \mc{S}$ decides in time $\poly(|I|)$ whether $s$ is a \emph{valid}
candidate solution for $I$. We denote with $\mc{S}_I \subseteq
\mc{S}$ the set of all valid candidate solutions for a fixed instance
$I$.
\item There exist two polynomial-time computable functions
$\pred$ and $\suc$ that define the edge set of $G$ as follows:
on input $I \in \mc{I}$ and $s \in \mc{S}_I$,
$\pred$ and $\suc$ return a valid candidate solution from $\mc{S}_I$
or $\bot$. Here, $\bot$ means
that $v$ has no predecessor/successor.
\item There is a polynomial-time algorithm that returns for each
instance $I$ a valid candidate solution $s \in \mc{S}_I$ with $\pred(s)
= \bot$. We call $s$ the \emph{standard source}.
\end{itemize}
Now, each instance $I \in \mc{I}$ defines a graph $G_I = (V,E)$ as follows. The
set of nodes $V$ is the set of all valid candidate solutions $\mc{S}_I$ and
there is a directed edge from $u$ to $v$ if and
only if $v = \suc(u)$ and $u = \pred(v)$. Clearly, each node in $G_I$ has
indegree and outdegree at most one. The relation $\mc{R}$ consists of all tuples
$(I,s)$ such that $s$ is a sink or source other than the standard source in $G_I$.

The definition of a \PPAD-problem suggests a simple algorithm, called the
\emph{standard algorithm}: start at the standard source and follow the path
until a sink is reached. This algorithm always finds a solution but the length
of the traversed path may be exponential in the size of the input instance.
\paragraph{Polyhedral Complexes and Subdivisions.}
We call a finite set of polyhedra $\mc{P}$ in $\R^d$ a \emph{polyhedral
complex} if and only if
(i) for all polyhedra $f \in \mc{P}$, all faces of $f$
are contained in $\mc{P}$; and (ii)
for all $f,f' \in \mc{P}$, the intersection $f \cap f'$ is a face of both.
Note that the first requirement implies that $\emptyset \in
\mc{P}$. Furthermore, we say $\mc{P}$ has \emph{dimension} $k$ if there 
exists some polyhedron $f \in \mc{P}$ with $\dim f = k$ and all other 
polyhedra in $\mc{P}$
have dimension at most $k$. We call $\mc{P}$ a \emph{polytopal complex}
if it is a polyhedral complex and all elements are polytopes. Similarly, 
we say
$\mc{P}$ is a \emph{simplicial complex} if it is a polytopal complex whose
elements are simplices. Finally, we say $\mc{P}$ \emph{subdivides} a set $Q
\subseteq \R^d$ if $\bigcup_{f \in \mc{P}} f = Q$. For more details,
see~\cite[Section~5.1]{Ziegler1995}.

\paragraph{Linear Programming.}
Let $A \in \R^{d \times n}$ be a matrix and $F$ a set of column
vectors from $A$. Then, we denote with $\ind{F} \subseteq [n]$ the set
of column indices in $F$ and for an index set $I \subseteq [n]$, we denote
with $A_I$ the submatrix of $A$ that consists of the columns indexed
by 
$I$. Similarly, for a vector $\cc \in \R^n$ and an index set 
$I \subset [n]$, we
denote with $\cc_I$ the subvector of $\cc$ with the coordinates
indexed
by $I$. Now, let $L'$ denote a system of linear equations
\begin{align*}
  L': A\xx =\bb,
\end{align*}
where $A \in \Q^{d \times n}$, $\bb \in \Q^d$ 
and $\rank(A) = k$.
By multiplying with the least common denominator,
we may assume in the following that $A\in \Z^{d \times n}$ and 
$\bb \in \Z^d$.
We call a set of $k$ linearly independent
column vectors $B$ of $A$ a \emph{basis} and we say that $A$ is
\emph{non-degenerate} if $k = d$ and for all bases $B$ of $A$, no 
coordinate of the corresponding solution $\xx_{\ind{B}}$ is $0$.
In particular, if $L'$ is non-degenerate, then $\bb$ is not contained 
in the linear span of any set of $d' < d$ column vectors from $A$ 
and hence if $d > n$, the linear system $L'$ has no solution. In the 
following, we assume that $L'$ is non-degenerate and that
$d\leq n$.

We denote with $L$ the linear program obtained by extending the 
linear system $L'$ with the constraints $\xx \geq \0$ and with a 
cost vector $\cc \in \Q^n$:
\[
  L:  \min \cc^T \xx  \text{ subject to } 
                 A\xx =\bb, \, \xx \geq \0.
\]
We say a set of column vectors $B$ is a \emph{basis} 
for $L$ if $B$ is a basis for $L'$. Let $\xx \in \R^n$ 
be the corresponding solution, i.e., let $\xx$ be
such that $A \xx = \bb$ and $x_i = 0$ for $i \in [n] \setminus
\ind{B}$. We call $\xx$ a \emph{basic feasible solution}, and
$B$ a \emph{feasible basis}, if $\xx \geq
\0$. 
Furthermore, we say $L$ is \emph{non-degenerate} if for
all feasible bases $B$, the corresponding basic feasible
solutions have strictly positive values in the coordinates
of $B$.
Now, let $R = [n] \setminus \ind{B}$ be the column indices 
not in $B$. The \emph{reduced cost vector} $\rr_{B,\cc} \in
\Q^{n-d}$ with respect to $B$ and $\cc$ is then defined as
\begin{equation}
\rr_{B,\cc} = \cc_{R} - \lt(A^{-1}_\ind{B} A_R\rt)^T \cc_\ind{B}.
\label{eq:redcosts}
\end{equation}
It is well-known that $B$ is optimal for $\cc$ if and only if
$\rr_{B,\cc}$ is non-negative in all
coordinates~\cite{MatousekGa2007}. For technical reasons, 
we consider in the
following the \emph{extended reduced cost vector} $\er_{B,\cc} \in \Q^n$
that has a $0$ in dimensions $\ind{B}$ and otherwise
equals $\rr_{B,\cc}$ to align the coordinates of the reduced cost vector with the column indices in $A$.
More formally, we set
\[
  \left(\er_{B,\cc} \right)_j = \begin{cases}
    0 & \text{if } j \in \ind{B}\text{, and} \\
    (\rr_{B,\cc})_{j'} & \text{otherwise,}
  \end{cases}
\]
where $j'$ is the rank of $j$ in $R$, that is, $(\rr_{B,\cc})_{j'}$ is the
coordinate of $\rr_{B,\cc}$ that corresponds to the $j'$th non-basis column with
column index $j$ in $A$.

Geometrically, the feasible solutions for the linear program $L$ define 
an $(n-d)$-dimensional polyhedron $\PP$ in $\R^n$.
Since $L$ is non-degenerate, $\PP$ is simple.
Let $f \subseteq \PP$ be a $k$-face of $\PP$. Then, $f$ has an
associated set $\supp{f} \subseteq [n]$ of $k$ column indices
such that $f$ consists precisely of the feasible solutions for
the linear program
$A_\supp{f} \xx' = \bb,\, \xx' \geq \0$, lifted to
$\R^n$ by setting the coordinates with indices not in $\supp{f}$
to $0$.
We call $\supp{f}$ the \emph{support} of $f$ and we say the columns in
$A_\supp{f}$ \emph{define} $f$.
Furthermore, 
for all subfaces $\down{f} \subseteq f$, we have $\supp{\down{f}} \subseteq
\supp{f}$ and in particular, all bases that define vertices of $f$ are
$d$-subsets of columns from $A_\supp{f}$.

Moreover, we say a nonempty face $f \subseteq \PP$ is \emph{optimal} 
for a cost vector $\cc$ if all points in $f$ are optimal for $\cc$. We 
can express this
condition using the reduced cost vector. Let $B$ be a basis for a vertex in
$f$. Then $f$ is optimal for $\cc$ if and only if
\[
    (\er_{B,\cc})_j  = 0 \text{~for $j \in \supp{f}$, and }
    (\er_{B,\cc})_j  \leq 0 \text{~otherwise.}
\]

\section{Overview of the PPAD-Formulation}
\label{sec:ppad}

We give a new constructive proof of
the cone version of the colorful \Caratheodory theorem based on 
Sperner's lemma. Using this, we can obtain a \PPAD-formulation
of \CCP, by adapting Papadimitriou's formulation of
Sperner's lemma as a \PPAD problem. 

Recall the statement of Sperner's lemma: let $\mc{S}$ be a simplicial
subdivision of the $d$-dimensional standard simplex 
$\DD^{d} = \conv(\e_1, \dots, \e_{d+1}) \subset
\R^{d+1}$, where $\e_i$ is the $i$th canonical basis vector. 
We call a function $\lambda$ that assigns to
each vertex in $\mc{S}$ a label from $[d+1]$ a \emph{Sperner labeling}
if for each vertex $\vv$ of $\mc{S}$ contained in  
$\conv(\e_{i_1}, \dots, \e_{i_k})$, we
have $\lambda(\vv) \in \{i_1, \dots, i_k\}$, for all 
$\{i_1, \dots, i_k\} \subseteq [d+1], k \in [d+1]$.
For a simplex $\sigma \in \mc{S}$, we set $\lambda(\sigma)$
to be the set of labels of the vertices of $\sigma$.
We call $\sigma$ \emph{fully-labeled} if $\lambda(\sigma) = [d+1]$.

\begin{figure}[htbp]
  \begin{center}
    \includegraphics[scale=0.7]{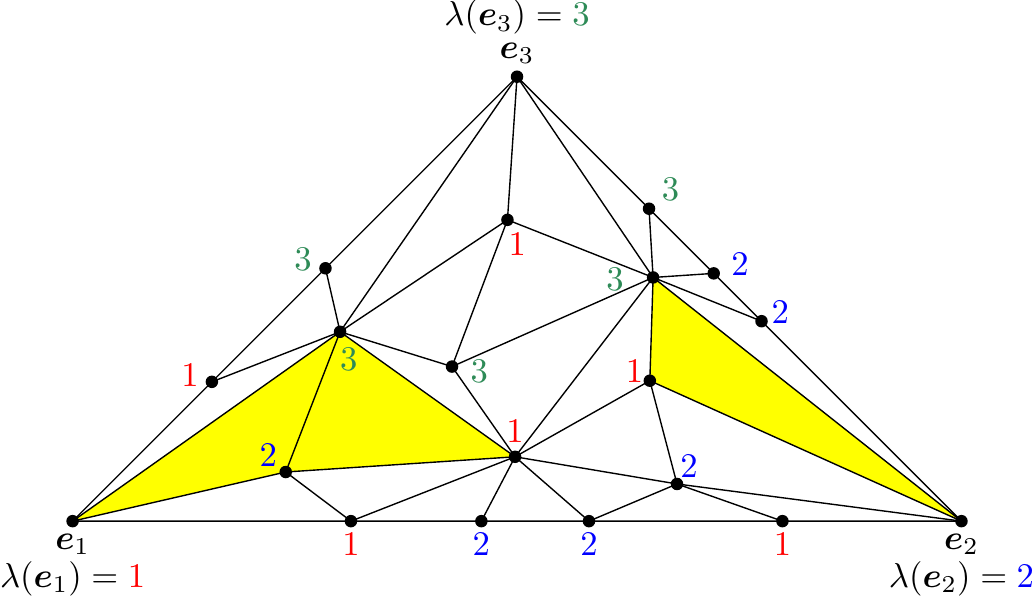}
  \end{center}
  \caption{An example of Sperner's lemma in two dimensions. The fully-labeled
  simplices are marked yellow.}
  \label{fig:sperner}
\end{figure}

\begin{theorem}[Strong Sperner's Lemma~\cite{Cohen1967}]
  \label{thm:sperner}
  The number of fully-labeled simplices is odd.
\end{theorem}

Now suppose we are given an instance $I = (C_1, \dots, C_d, \bb)$ of
(the cone version of)
\CCP, where $\bb \in \R^d$, $\bb \neq \0$, and each $C_i \subset
\Q^d$, $i  \in [d]$, ray-embraces $\bb$. In Section~\ref{sec:app:eqccp}, we
show that we can assume w.l.o.g.\ that each set $C_i$ has size $d$.
We now describe how to define a simplicial complex $\mc{S}$
and a Sperner labeling $\lambda$ for $I$ such that a fully labeled
simplex will encode a colorful choice that contains the 
vector $\bb$ in its positive span.

In the following, we call $\R^d$ the \emph{parameter space} and a
vector $\mm \in \R^d$ a \emph{parameter vector}.
We define a family of linear programs $\{\LC_\mm \mid \mm \in
\R^d\}$, where each linear program $\LC_\mm$ has the same constraints
and differs only in its cost vector $\cc_\mm$. The cost vector
$\cc_\mm$ is defined by a linear function in $\mm \in \R^d$.
Let $A = ( C_1\; C_2 \;\dots\; C_d ) \in \Q^{d \times d^2}$ be
the matrix that has the vectors from $C_1$ in the first $d$ columns,
the vectors from $C_2$ in the second $d$ columns, and so on.
Then, we denote with $\LC_{\mm}$ the linear program
\begin{equation}\label{eq:lp}
  \LC_{\mm}:
      \min \cc^T_{\mm} \xx,\,
      \text{ subject to } 
                      A \xx  = \bb, 
                      \xx  \geq \0,
\end{equation}
and we denote with $\PC \subset \R^{d^2}$ the polyhedron that is 
defined by the linear system $\LC$.
We can think of the $i$th coordinate of the parameter vector 
$\mm \in \R^d$ as the weight of color $i$, i.e., the costs of columns 
from $A$ with color $i$ decrease if $(\mm)_i$ increases. 
To each face $f$ of $P$, we assign the set of parameter vectors
$\Phi(f) \subset \R^d$ such that for all $\mm \in \Phi(f)$, the face
$f$ is optimal for the linear program $\LC_\mm$ that has $\LC$ as
constraints and $\cc_\mm$ as cost vector. We call $\Phi(f)$ the
\emph{parameter region} of $f$.  The cost vector is designed to
control the colors that appear in the support of optimal faces for a
specific subset of parameter vectors.  Let $\MM = \left\{ \mm \in
\R^d \midd \mm \geq \0,\, \|\mm\|_\infty =1 \right\}$ denote the
faces of the unit cube in which at least one coordinate is set to
$1$. Then, no face $f$ that is assigned to a parameter vector $\mm
\in \MM$ with $(\mm)_{i^\times}=0$ has a column from $A$ with color
$i^\times$ in its defining set $A_\supp{f}$. This property will become
crucial when we define a Sperner labeling later on.
 Now, we define a  polyhedral subcomplex $\FF$ of $\PC$ that consists 
 of all faces $f$ of $\PC$
such that $\Phi(f) \cap \MM \neq \emptyset$.  Furthermore, 
the
intersections of the parameter regions with $\MM$ induce a polytopal
complex $\QQ$ that is in a dual relationship to $\FF$. 
By performing a central projection with the origin as
center of $\QQ$ onto the standard simplex $\DD^{d-1}$, we obtain a
polytopal subdivision $\QQ_\DD$ of $\DD^{d-1}$.
To get the desired \emph{simplicial} subdivision of
$\DD^{d-1}$, we take the 
barycentric subdivision $\sd \QQ_\DD$
of $\QQ_\DD$.

We construct a Sperner labeling $\lambda$ for $\sd \QQ_\DD$
as follows:
let $\vv$ be a vertex in $\sd \QQ_\DD$, and let $f$
be the face of $\FF$ that corresponds to $\vv$. Then,
we set $\lambda(\vv) = i$ if the
$i$th color appears most often in the support of $f$.
The color controlling property of the
cost function $c_\mm$ then implies that $\lambda$ is a Sperner labeling. 
Furthermore, 
using the properties of the barycentric subdivision and the 
correspondence between $\QQ_\DD$ and $\FF$, we can 
show that one vertex of a fully-labeled
$(d-1)$-simplex in $\sd \QQ_\DD$ encodes a colorful feasible basis
of the \CCP instance $I$. This
concludes a new constructive proof of the colorful \Caratheodory
theorem using Sperner's lemma.

To show
that \CCP is in \PPAD however, we need to be able to traverse $\sd
\QQ_\DD$ efficiently. For this, we introduce a combinatorial encoding
of the simplices in $\QQ_\DD$ that represents neighboring simplices
in a similar manner.  
Furthermore, we describe how 
to generalize the orientation used in the \PPAD 
formulation of 2D-Sperner~\cite{Papadimitriou1994} to our
setting.
This finally shows that \CCP is in \PPAD.

To ensure that the complexes that appear in our algorithms are
sufficiently generic, we prove several perturbation lemmas
that give a deterministic way of achieving this.
Our \PPAD-formulation also shows that the special case of \CCP
involving two colors can be solved in polynomial time. Indeed,
we will see that in this case the polytopal complex $\QQ_\DD$
can be made $1$-dimensional. Then,
binary search can be used to find a fully-labeled simplex in
$\QQ_\DD$. In order to prove that
the binary search terminates after a polynomial number of steps,
we use methods similar to our perturbation techniques to
obtain a bound on the length of the $1$-dimensional fully-labeled simplex.

\section{The Colorful \Caratheodory Problem is in PPAD}
\label{sec:ccpppad}

As before, let $I=(C_1,\dots,C_d,\bb)$ denote an instance for the cone version of
\CCP. Our formulation of \CCP as a \PPAD-problem requires $I$ to be in
general position. In particular, we assume that (P1) all color classes
$C_i \subset \Z^d$ consist of $d$ points and all points have integer
coordinates. Furthermore, we assume that (P2) there exist no subset $P
\subset \bigcup_{i=1}^d C_i$ of size $d-1$ that ray-embraces $\bb$. We
show in Section~\ref{sec:app:eqccp} how to ensure the
properties by an explicit deterministic perturbation of polynomial
bit-complexity.

\subsection{The Polytopal Complex}
\label{sec:ppad:para}
 Let $N = d!m^d$, where
$m$ is the largest absolute value that appears in $A$ and $\bb$
(see~Lemma~\ref{lem:lpsol}).
Then, we define $\cc_{\mm} \in \R^{d^2}$ as
\begin{equation}\label{eq:costs}
  (\cc_{\mm})_j = 1 + \left(1 - (\mm)_i\right) d N^2 + \eps^j,
\end{equation}
where $j \in [d^2]$,
$i$ is the color of the $j$th column in $A$, and $0 < \eps  \leq N^{-3}$ 
is a suitable perturbation that ensures non-degeneracy of the
reduced costs (see~\cite{chvatal1983}).
As stated in the overview, the cost function
controls the colors in the support of the optimal faces for 
parameter vectors in $\MM$. The proof of the following lemma
can be found in Section~\ref{sec:app:polycompl}.

\begin{lemma}\label{lem:colors}
Let $i^\times \in [d]$ be a color and let $\mm \in \MM$ be 
a parameter vector
with $\mm_{i^\times} = 0$. Furthermore, let $B^\star$ be 
an optimal feasible
basis for $\LC_\mm$. Then, $B^\star \cap C_{i^\times} = \emptyset$.
\end{lemma}

We denote for a face 
$f \subseteq \PC$, $f \neq \emptyset$,
with
$
  \Phi(f) = \big\{\mm \in \R^d \mid \text{$f$ is optimal for
      $L_{\mm}$}\big\}
$
the set of all parameter vectors for which $f$ is optimal. We call 
this the
\emph{parameter region} for $f$. Using the reduced cost vector, we can
express $\Phi(f)$ as solution space to the
following linear system, where $B$ is a feasible basis of some vertex of
$f$ and the $d$ coordinates of the parameter vector $\mm$ are the variables:
\begin{equation}
  \LR_{B,f}:
 (\er_{B,\cc_\mm})_j  = 0 \text{\ for } j \in \supp{f} \setminus
 \ind{B} \text{ and }
(\er_{B,\cc_\mm})_j  \leq 0 \text{\ for }\big[d^2\big] \setminus \supp{f}.
\end{equation}
Then, we define $\FF$ as the set of all faces that are optimal for some
parameter vector in $\MM$:
\[
  \FF = \left\{ f \midd \text{$f$ is a face of $\PC$},\, 
  \Phi(f) \cap \MM \neq
  \emptyset \right\}.
\]
By definition, $\FF \cup \{\emptyset\}$ is a
polyhedral subcomplex of $\PC$. 
The intersections of the parameter regions with faces of $\MM$ induce 
a subdivision
$\QQ$ of $\MM$:
\[
  \QQ = \left\{ \Phi(f) \cap g \midd f \in \FF,\,
  \text{$g$ is a face of $\MM$} \right\}.
\]
In Section~\ref{sec:app:polycompl}, we show that $\QQ$ is a
$(d-1)$-dimensional polytopal complex.
Next, we construct $\QQ_\DD$ through a central projection with 
the origin as center of
$\QQ$ onto the $(d-1)$-dimensional standard simplex $\DD \subset \R^d$. 
It is
easy to see that this projection is a bijection. For a
parameter vector $\mm \in \R^d$, we denote with
$
\DD(\mm) =  \mm/\|\mm\|_1
$
its projection onto $\DD$. Similarly, we
denote with 
$
\MM(\mm) =  \mm/\|\mm\|_\infty
$
the
projection of $\mm$ onto $\MM$ and we use the same notation to denote the
element-wise projection of sets. Then, we can write the projection 
$\QQ_\DD$ of
$\QQ$ onto
$\DD$ as $\QQ_\DD = \{\DD(q) \mid q \in \QQ\}$.
Furthermore, let $\SS = \{ \Delta(g) \mid \text{$g$ is a face of $\MM$}\}$
denote the projections
of the faces of $\MM$ onto $\DD$. For $f \in \FF$, let
$\Phi_\DD(f) = \DD(\Phi(f) \cap \MM)$ denote the projection of all 
parameter vectors in
$\MM$ for which $f$ is optimal onto $\Delta$. Please refer to
Table~\ref{tab:notation} on Page~\pageref{tab:notation} for an
overview of the current and future notation.
The following results are proved in Section~\ref{sec:app:polycompl}.

\begin{lemma}\label{stm:face_dim}\label{stm:unique}
Let $q \neq \emptyset$ be an element from
$\QQ_\DD$. Then, there exists unique pair
$(f, g)$ where $f$ is a face of  $\FF$
and $g$ is a face of $\SS$ such that $q = \Phi_{\Delta}(f) \cap g$.
Moreover, $q$ is a simple polytope of dimension $\dim g - \dim f$ and, 
if $\dim q > 0$, the set of facets of $q$ can be written as
\[
  \left.\Big\{ \Phi_{\Delta}\left(f\right) \cap \facet{g} \neq \emptyset \midd
  \text{$\facet{g}$ is a facet of $g$} \Big\}\right. \cup
  \left\{ \Phi_{\Delta}\left(\up{f}\right) \cap g \neq \emptyset \midd \text{$f$ is a
  facet of $\up{f} \in \FF$}\right\}.
\]
\end{lemma}

\begin{lemma}\label{stm:polycompl}
The set $\QQ_\DD$ is a $(d-1)$-dimensional polytopal
complex that decomposes $\Delta$. \qed
\end{lemma}

\subsection{The Barycentric Subdivision}
\label{sec:ppad:sd}

The \emph{barycentric subdivision}~\cite[Definition~1.7.2]{Matousek2008} 
is a
well-known method to subdivide a polytopal complex into
simplices. We define $\sd \QQ_\DD$ 
as the set of  all simplices $\conv(\vv_0,\dots,\vv_k)$, $k \in [d]$, 
such that
there exists a chain $q_0 \subset \dots \subset q_k$ of polytopes in
$\QQ_\DD$ with $\dim q_{i-1} < \dim q_i$ and such that $\vv_i$ is 
the barycenter of
$q_i$ for $i \in [k]$. We define the label of a
vertex $\vv \in \sd \QQ_\DD$ as follows. By Lemma~\ref{stm:unique}, there
exists a
unique pair $f \in \FF$ and $g \in \SS$ with $\vv = \Phi_\DD(f) \cap g$. Then, the
label $\lambda(\vv)$ of $\vv$ is defined as
\begin{equation}
  \label{eq:ppad:labeling}
 \lambda(\vv) = \argmax_{i \in [d]} \lt|\ind{C}_i \cap \supp{f}\rt|.
\end{equation}
In case of a tie, we take the smallest $i \in [d]$ that achieves the
maximum. Lemma~\ref{lem:colors} implies that $\lambda(\cdot)$ is a Sperner 
labeling
of $\sd \QQ_\DD$. 
In fact, $\lambda$ is 
a Sperner labeling for any fixed simplicial subdivision of $\Delta$.
Now, Theorem~\ref{thm:sperner} guarantees the existence of a $(d-1)$-simplex 
$\sigma
\in \sd \QQ_\DD$ whose vertices have all $d$ possible labels. 
The next lemma
shows that then one of the vertices of $\sigma$ defines a solution to the
$\CCP$ instance. Here, we use specific properties of the barycentric
subdivision.
\begin{lemma}\label{lem:fullycolored_colbasis}
Let $\sigma \in \sd \QQ_\DD$ be a fully-labeled
$(d-1)$-simplex and let $\vv_{d-1}$ denote the vertex of $\sigma$ that 
is the
barycenter of a $(d-1)$-face $q_{d-1} = \Phi_\DD(f_{d-1}) \cap g_{d-1} \in
\QQ_\DD$, where $f_{d-1} \in \FF$ and $g_{d-1} \in \SS$.  Then, the 
columns from
$A_{\supp{f_{d-1}}}$ are a colorful choice that ray-embraces $\bb$.
\end{lemma}

Our discussion up to now already yields a new Sperner-based proof 
of the colorful \Caratheodory theorem.
However, in order to show that $\CCP \in \PPAD$, we need to replace
the invocation of Theorem~\ref{thm:sperner} by a \PPAD-problem. Note that it is 
not possible to use the formulation of Sperner
from~\cite[Theorem~2]{Papadimitriou1994} directly, since
it is defined for a fixed simplicial subdivision of the standard 
simplex. In our
case, the simplicial subdivision of $\Delta$ depends on the input instance.
In the following, we generalize the \PPAD formulation of
Sperner in \cite{Papadimitriou1994} to $\QQ_\DD$ by mimicking the proof of
Theorem~\ref{thm:sperner}. For this, we need to be able to find simplices
in $\sd \QQ_\DD$ that share a given facet. We begin with a
simple encoding of simplices in $\sd \QQ_\DD$ that allows us to
solve this problem completely combinatorially.

We first show how to encode a polytope $q \in \QQ_\DD$. By 
Lemma~\ref{stm:unique},
there exists a unique pair of faces $f \in \FF$ and $g \in S$ such 
that $q = \Phi_\DD(f) \cap g$. Since $\MM(g)$ is a face of the unit 
cube, the value of $d - \dim g$ coordinates in $\MM(g)$ is fixed 
to either $0$ or $1$. Let $I_j
\subseteq [d]$,
$j=0,1$, denote the indices of the coordinates that are fixed to $j$.
Then, the encoding of $q$ is defined as
$
  \en{q} = \lt(\supp{f}, I_0, I_1\rt)$.
We use this to define an encoding of the simplices in $\QQ_\DD$ as follows.
Let $\sigma \in \QQ_\DD$ be a $k$-simplex and let $q_0 \subset\dots
\subset q_k$ be the corresponding face chain in $\QQ_\DD$ such that the 
$i$th vertex of $\sigma$ is the barycenter of $q_i$. Then, the 
encoding $\en{\sigma}$ is defined as
\begin{equation}\label{eq:enc:simplex}
  \en{\sigma} = \lt(\en{q_0},\dots,\en{q_k}\rt).
\end{equation}

In the proof of Theorem~\ref{thm:sperner}, we traverse only a subset of simplices
in the simplicial subdivision, namely $(k-1)$-simplices that are 
contained in the face $\DD_{[k]} = \conv\{\e_i \mid i \in [k]\}$ of $\DD$ for $k \in [d]$.
Let  
$
  \FS_k = \left\{ \sigma \in \sd \QQ_\DD \midd \dim(\sigma) = k-1,\ \sigma
  \subseteq \DD_{[k]} \right\}
$
denote the set
of $(k-1)$-simplices in $\sd \QQ_\DD$ that are contained in the $(k-1)$-face,
where $k \in [d]$, and let $\FS = \bigcup_{k=1}^d
\FS_k$ be the collection of all those simplices. In the following, we give a
precise characterization of the encodings of the simplices in $\FS_k$.
For two disjoint index sets $I_0,I_1 \subseteq [d]$, we denote with
$g(I_0,I_1)= \left\{ \mm \in
\MM \midd j=0,1,\, (\mm)_i = j \text{ for $i \in I_j$}\right\}$ the
face of
$\MM$ that we obtain by fixing the coordinates in dimensions $I_0\cup I_1$.
Let now $T=\lt(Q_0,\dots,Q_{k-1}\rt)$, $k \in [d-1]$, be a tuple, where $Q_i =
\lt(S^{(i)}, I^{(i)}_0,
I^{(i)}_1\rt)$, $S^{(i)} \subset \lt[d^2\rt]$, and $I^{(i)}_0, I^{(i)}_1$ are
disjoint subsets of $[d]$ with $I^{(i)}_1 \neq \emptyset$ for $i \in [k-1]_0$.
We say $T$ is \emph{valid} if and only if $T$ has the following
properties.
\begin{enumerate}[label=(\roman{enumi})]
  \item\label{valid:ini} We have $I^{(k-1)}_0 = [d] \setminus [k]$,
  $\lt|I^{(k-1)}_1\rt| = 1$,
  and the columns in $A_{S^{(k-1)}}$ are a feasible basis for a vertex $f$.
  Moreover, the intersection $\Phi(f) \cap g\lt(I^{(k-1)}_0 \cup
  I^{(k-1)}_1\rt)$ is nonempty.
  \item\label{valid:facet} For all $i \in [k-1]$, we either have
    \begin{enumerate}[label=(\roman{enumi}.\alph{enumii})]
      \item\label{valid:facet:f} $I^{(i-1)}_0 = I^{(i)}_0$,
        $I^{(i-1)}_1 = I^{(i)}_1$, and $S^{(i-1)} = S^{(i)} \cup \left\{ a_{i-1}
        \right\}$ for some index $a_{i-1} \in \lt[d^2\rt] \setminus S^{(i)}$,
      \item\label{valid:facet:g} or $S^{(i-1)} = S^{(i)}$ and there is an index
        $j_{i-1} \in [d] \setminus \lt(I^{(i)}_0 \cup I^{(i)}_1\rt)$ such that
        either $I^{(i-1)}_0 = I^{(i)}_0$ and $I^{(i-1)}_1 = I^{(i)}_1 \cup
        \left\{ j_{i-1} \right\}$, or $I^{(i-1)}_1 = I^{(i)}_1$ and $I^{(i-1)}_0
        = I^{(i)}_0 \cup \left\{ j_{i-1} \right\}$.
    \end{enumerate}
\end{enumerate}

\begin{lemma}\label{lem:enc_bij}
  For $k \in [d]$, the function $\en{\cdot}$ restricted to the simplices in
  $\FS_k$ is a bijection from $\FS_k$ to the set of valid $k$-tuples.
\end{lemma}

Using our characterization of encodings as valid tuples, it becomes 
an easy task to check whether a given candidate encoding corresponds 
to a simplex in $\FS$.

\begin{lemma}\label{lem:enc_verify}
Let $T=\lt(Q_0,\dots,Q_{k-1}\rt)$, $k \in [d-1]$, be a tuple,
where $Q_i =
\lt(S^{(i)}, I^{(i)}_0,
I^{(i)}_1\rt)$, $S^{(i)} \subset \lt[d^2\rt]$, and $I^{(i)}_0, I^{(i)}_1$ are
disjoint subsets of $[d]$ with $I^{(i)}_1 \neq \emptyset$ for $i \in [k-1]_0$.
Then, we can check in polynomial time whether $T$ is a valid $k$-tuple.
\end{lemma}

In Section~\ref{sec:app:barycentric}, we show that simplices in $\FS$ that
share a facet have similar encodings that differ only in one element
of the encoding tuples.
Using this fact, we can  traverse $\FS$ efficiently by manipulating the
respective encodings.

\begin{lemma}\label{lem:enc_algs}
Let $\sigma \in\FS_k$ be a simplex and let 
$q_0 \subset \dots \subset q_{k-1}$ be the 
corresponding face chain in $\QQ_\DD$ such 
that the $i$th vertex $\vv_i$ of $\sigma$ is 
the barycenter of $q_i$, where $k \in[d]$ and $i \in
[k-1]_0$. Then, we can solve the following problems 
in polynomial time:
(i) Given $\en{\sigma}$ and $i$, compute the encoding of the simplex
$\sigma' \in \FS_k$ that shares the facet
$\conv\left\{ \vv_j \midd j \in [k-1]_0,\, j \neq i \right\}$ with 
$\sigma$ or
state that there is none; (ii) Assuming that $k<d$ and 
given $\en{\sigma}$, compute the encoding of
the simplex $\up{\sigma} \in \FS_{k+1}$ that has $\sigma$ as facet;
and (iii) Assuming that $k>1$ and given $\en{\sigma}$, compute 
the encoding of the simplex $\down{\sigma} \in \FS_{k-1}$ that 
is a facet of $\sigma$ or state that there is none.
\end{lemma}

\subsection{The PPAD graph}
\label{sec:ppad:graph}

Using our tools from the previous sections, we now describe the \PPAD
graph $G=(V,E)$ for the \CCP instance. 
The definition of $G$ follows mainly the ideas
from the formulation of Sperner as a
PPAD-problem~\cite[Theorem~2]{Papadimitriou1994} and the proof of
Theorem~\ref{thm:sperner}.

The graph has one node per simplex in $\FS$ that has all labels or 
all but the largest possible label. That is, we have one node for 
each $(k-1)$-simplex $\sigma$ in
$\FS_k$ with $[k-1] \subseteq \lambda(\sigma)$. Two
simplices are connected by an edge if one simplex is the facet of 
the other or
if both simplices share a facet that has all but the largest 
possible label.
More formally, 
for $k \in [d]$,
we set
$
V_k = \left\{ \en{\sigma} \midd \sigma \in \FS_k,\, [k-1] \subseteq
\lambda(\sigma) \right\}
$,
the set of all encodings for $(k-1)$-simplices in $\FS_k$ whose vertices 
have
all or all but the largest possible label. 
Then, $V$ is the union of all $V_k$
for $k \in [d]$.
There are
two types of edges: edges within a set $V_k$, $k \in [d]$, and
edges connecting nodes from $V_k$ to nodes in $V_{k-1}$
and $V_{k+1}$. Let $\en{\sigma}, \en{\sigma'}$ be two vertices in
$V_k$ for some $k \in [d]$. Then, there is an edge between
$\en{\sigma}$ and $\en{\sigma'}$ if the
encoded simplices $\sigma,\, \sigma' \in \FS_k$ share a facet
$\facet{\sigma}$ with $\lambda(\facet{\sigma}) = [k-1]$, i.e., both
simplices are connected by a facet that has all but the largest possible
label. Now, let $\en{\sigma} \in V_k$ and $\en{\sigma'} \in
V_{k+1}$ for some $k \in [d-1]$. Then, there is an edge between
$\en{\sigma}$ and $\en{\sigma'}$ if $\lambda(\sigma) = [k]$ and
$\sigma$ is a facet of $\sigma'$.
In the next lemma, we show that 
$G$ consists only of paths and cycles. Please see
Section~\ref{sec:app:ppadgraph} for the proof.

\begin{lemma}\label{stm:deg}
Let $\en{\sigma} \in V$ be a node. If $\en{\sigma} \in V_1$ or
$\en{\sigma} \in V_d$ with $\lambda(\sigma) = [d]$, then $\deg \en{\sigma}
= 1$. Otherwise, $\deg \en{\sigma}  = 2$.
\end{lemma}

This already shows that $\CCP \in \PPA$. By generalizing the
orientation from \cite{Papadimitriou1994} to our setting, we obtain a
function $\dir$ that orients the edges of $G$ such that only vertices
with degree one in $G$ are sinks or sources in the oriented graph. In
Section~\ref{sec:app:ppadgraph}, we show how to compute this
function in polynomial time.
This finally yields our main result.

\begin{theorem}\label{thm:ppad}
\CCP, \Centerpoint, \Tverberg, and \SimCenter are in $\PPAD \cap \PLS$.
\end{theorem}
\begin{proof}
We give a formulation of \CCP as \PPAD-problem. See Section~\ref{sec:app:pls}
for a formulation of \CCP as \PLS-problem. Using the
classic proofs discussed in Section~\ref{sec:app:reductions},
this then also implies the statement for the other problems.

The set of problem instances $\mc{I}$ consists of all tuples
$I=(C_1,\dots,C_d,\bb)$, where $d \in \N$, the set $C_i \subset \Q^d$
ray-embraces $\bb \in \Q^d$ and $\bb \neq \0$. Let $I^\apx=
(C^\apx_1,\dots,C^\apx_d,\bb^\apx)$ denote then the \CCP instance that we
obtain by applying our perturbation techniques to $I$ (see
Section~\ref{sec:app:eqccp}). Then, $I^\apx$ has the general position
properties
(P1) and (P2). The set of candidate
  solutions $\mc{S}$ consists of all tuples $(Q_0,\dots,Q_{k-1})$,
  where $k \in \N$ and $Q_i$ is a tuple 
  $\lt(S^{(i)},I^{(i)}_0,I^{(i)}_1\rt)$
  with $S^{(i)},I^{(i)}_0,I^{(i)}_1 \subset \N$. Furthermore, 
  $\mc{S}$ contains
  all $d$-subsets $C \subset \Q^d$ for $d \in \N$.
  We define the set of valid
  candidate solutions $\mc{S}_I$ for the instance $I$ to be the set of all valid
  $k$-tuples with
  respect to the instance $I^\apx$ and the set of all colorful choices with
  respect to $I$ that
  ray-embrace $\bb$, where $k \in [d]$. Let $s \in \mc{S}$ be a candidate
  solution. If it is a tuple, we first use the algorithm from
  Lemma~\ref{lem:enc_verify} to check in polynomial time in the length of
  $I^\apx$ and hence in the length of $I$ whether $s \in \mc{S}_I$.
  If affirmative, we check whether the simplex has all or all but the 
  largest possible label. Using the encoding, this can be carried 
  out in polynomial time.  If $s$ is a set of points, we can determine
  in polynomial time with linear programming whether the points in $s$
  ray-embrace $\bb$.

  We set as standard source the $0$-simplex $\{\e_1\}$. We can assume
  without loss of generality that $\{\e_1\}$ is a source (otherwise we
  invert the orientation).

  Given a valid candidate solution $s \in \mc{S}_I$, we compute its
  predecessor and successor with the algorithms from
  Lemma~\ref{lem:enc_algs} and the orientation function discussed
  above, with one
  modification: if a node $s \in V$ is a source different from the
  standard source in the graph $G$, it encodes by
  the above discussion
  a colorful choice $C^\apx$ that
  ray-embraces $\bb^\apx$.  Let $C$ be the corresponding colorful
  choice for $I$ that ray-embraces $\bb$.  Then, we set the
  predecessor of $s$ to $C$. The properties of our perturbation 
  ensure that we can compute $C$ in
  polynomial time. Similarly, if $s$ is a sink in $G$, we set its
  successor to the corresponding solution for the instance $I$.
\end{proof}

\section{A Polynomial-Time Case}
\label{sec:halfhalf}
We show that for a special case of \CCP, our formulation
of \CCP as a \PPAD problem has algorithmic implications. 
Let $C_1,C_2 \in
\R^d$ be two color classes and let $C \subseteq
C_1 \cup C_2$ be a set. We call $C$ an 
\emph{$(k,d-k)$-colorful
choice} for $C_1$ and $C_2$ if there are two subsets
$C'_1 \subseteq C_1$, $C_2' \subseteq C_2$ with $|C'_1| \leq k$
and $|C'_2| \leq d-k$.
Now, given two color classes $C_1, C_2$ that each
ray-embrace a point $\bb \in \R^d$ and a number $k \in [d]_0$, 
we want to find an $(k,d-k)$-colorful choice that ray-embraces $\bb$.
It is a straightforward
consequence of the colorful \Caratheodory theorem that such 
a colorful choice always exists.

Using our techniques from Section~\ref{sec:ccpppad}, we present a weakly
polynomial-time algorithm for this case. As described in
Section~\ref{sec:ppad:para}, we construct
implicitly a $1$-dimensional polytopal complex, where at least one edge
corresponds to a solution. Then, we apply binary search to
find this edge. Since the length of the edges can be exponentially 
small in the
length of the input, this results in a weakly polynomial-time algorithm.

\begin{theorem}
  Let $\bb \in \Q^d$ be a point and let $C_1, C_2 \subset \Q^d$ be 
  two sets of size $d$ that ray-embrace $\bb$.  Furthermore, 
  let $k \in [d-1]$ be a parameter. Then, there is an algorithm that 
  computes a $(k, d-k)$-colorful choice $C$ that ray-embraces $\bb$ 
  in weakly-polynomial time.
\end{theorem}

For Sperner's lemma, it is well-known that a fully-labeled simplex 
can be found if there are only two labels by binary search. 
Essentially, this is also what the presented algorithm
does: reducing the problem to Sperner's lemma and then applying 
binary search to find the right simplex. Since the computational 
problem Sperner is \PPAD-complete even for $d=2$, a polynomial-time 
generalization of this approach to three colors must use specific 
properties of the colorful \Caratheodory instance
under the assumption that no \PPAD-complete problem can be solved in
polynomial time.

\section{Conclusion}
\label{sec:conclusion}

We have shown that \CCP lies in the intersection of \PPAD and \PLS.
This also immediately implies that several
illustrious problems associated with \CCP, such as finding centerpoints
or Tverberg partitions, belong to $\PPAD \cap \PLS$.

Previously, the intersection $\PPAD \cap \PLS$ has
been studied in the context of \emph{continuous
local search}: Daskalakis and Papadimitriou~\cite{DaskalakisPa11}
define a subclass $\CLS \subseteq \PPAD \cap \PLS$ that 
``captures a particularly benign kind of local optimization''.
Daskalakis and Papadimitriou describe several interesting
problems that lie in \CLS but are not known to
be solvable in polynomial time.
Unfortunately, our results do not show that
\CCP lies in \CLS, since we reduce \CCP in $d$ dimensions to
Sperner in $d-1$ dimensions, and since Sperner is not known 
to be in \CLS. Indeed, if Sperner's lemma could
be shown to be in \CLS, this would imply that
$\PPAD = \CLS \subseteq \PLS$, solving a major
open problem. Thus, showing that \CCP lies
in \CLS would require fundamentally new ideas, maybe
exploiting the special structure of the resulting 
Sperner instance. On the other hand,
it appears that Sperner is a more difficult problem than
\CCP, since Sperner is \PPAD-complete for every fixed dimension
larger than $1$, whereas \CCP becomes hard only in unbounded
dimension.
On the positive side, our perturbation results show that 
a polynomial-time algorithm for \CCP, even under strong
general position assumptions, would
lead to polynomial-time algorithms for several well-studied problems 
in high-dimensional computational geometry.

Finally, it would also be interesting to find further special
cases of \CCP that are amenable to polynomial-time solutions.
For example, can we extend our algorithm for two color classes to
\emph{three} color classes? We expect this to be difficult, due to
an analogy between 1D-Sperner, which is in \textsf{P},  and 2D-Sperner, 
which
is \PPAD-complete. However, there seems to be
no formal justification for this intuition.
\bibliographystyle{plain}

\clearpage
\appendix

\begin{table}\small
\begin{center}
\renewcommand{\arraystretch}{1.3}
\begin{tabularx}{\textwidth}{|c|X|}
  \hline
    \textbf{Symbol} &
    \multicolumn{1}{c|}{\textbf{Definition}} 
  \\ \hline \hline
    $C_i$ & The $i$th color class. The $d$-set $C_i \subset \R^d$
    ray-embraces $\bb$.
  \\ \hline
    $A$ & The $(d\times d^2)$-matrix with $C_1$ as first $d$ columns, $C_2$ as second
    $d$ columns, and so on.
  \\ \hline
    $c_\mm$ & The cost vector parameterized by a parameter vector $\mm \in \R^d$.
    See (\ref{eq:costs}).
  \\ \hline
    $\LC$; $\LC_\mm$ & $\LC$ refers to the linear system $A \xx = \bb$,
    $\xx\geq \0$ (see~\ref{eq:lp}). $\LC_\mm$ denotes the linear program $\max
    c_\mm^T \xx$ s.t.\ $\LC$.
  \\ \hline
    $\PC$ & The polytope defined by $\LC$.
  \\ \hline
    $f$; $\supp{f}$; $\ind{B}$ & For a face $f \subseteq \PC$, we denote with
    $\supp{f}$ the
    indices of the columns in $A$ that define it. For a set of
    columns $B$ of $A$, we denote with $\ind{B}$ the indices of these columns.
  \\ \hline
    $\Phi(f)$; $\LR_{B,f}$ & For a face $f$ of $\PC$, $\Phi(f)$ denotes the set
    of parameter vectors $\mm \in \R^d$ such that $f$ is optimal for $\LC_\mm$.
    The set $\Phi(f)$ can be described as the solution space to the linear
    system $\LR_{B,f}$, where $B$ is a feasible basis of a vertex of $f$.
  \\ \hline
    $\MM$ & The set $\MM$ contains all faces from the unit cube in $\R^d$ that set at
    least one coordinate to $1$. Parameters from $\MM$ control the colors of the
    defining columns of optimal faces (see Lemma~\ref{lem:colors}).
  \\ \hline
    $\FF$ & The set of faces $f$ of $\PC$ of that are optimal for some parameter vector
    in $\MM$, i.e., the set of faces $f$ with $\Phi(f) \cap \MM \neq
    \emptyset$. $\FF$ is a $(d-1)$-dimensional polyhedral complex.
  \\ \hline
    $\QQ$ & The $(d-1)$-dimensional polytopal complex that consists of all
    elements $q = \Phi(f) \cap g$, where $f \in \FF$ and $g$ is a
    face of $\MM$.
  \\ \hline
    $\DD$; $\DD_{[k]}$ & $\DD$ denotes the $(d-1)$-dimensional standard
    simplex and $\DD_{[k]}$ denotes the face $\conv\{\e_i \mid i \in [k] \}$ of
    $\DD$.
  \\ \hline
    $\SS$ & The set $\SS$ contains the central
    projections of the faces of $\MM$ onto $\DD$ with the origin as center.
  \\ \hline
    $\Phi_\DD$; $\QQ_\DD$ & $\Phi_\DD(f)$ denotes the central projection of
    $\Phi(f)\cap
    \MM$ onto $\DD$ with center $\0$. The $(d-1)$-dimensional polytopal complex
    $\QQ_\DD$
    consists of the projections of the elements in $\QQ$ onto $\DD$. Each
    element $q$ of
    $\QQ_\DD$ can be uniquely written as $q = \Phi_\DD(f) \cap g$, where $f \in
    \FF$ and $g \in \SS$.
  \\ \hline
    $\lambda$ & The labeling function, see (\ref{eq:ppad:labeling}).
  \\ \hline
    $\FS$; $\FS_k$; $\en{\sigma}$ & The set $\FS_k$, $k \in [d]$, consists of all
    $(k-1)$-simplices in $\sd \QQ_\DD$ that are contained in the face
    $\DD_{[k]}$ of $\DD$. The set $\FS$ is the union of all $\FS_k$. For a
    simplex $\sigma \in \FS$, we denote with $\en{\sigma}$ its combinatorial
    encoding (see (\ref{eq:enc:simplex})).
  \\ \hline
\end{tabularx}
\end{center}
\caption{Notation reference.}
\label{tab:notation}
\end{table}

\clearpage

\section{Polynomial-Time Reductions to the Colorful \Caratheodory Problem}
\label{sec:app:reductions}

We begin by presenting the proofs of the centerpoint theorem, Tverberg's theorem, and
the first selection lemma that use the colorful \Caratheodory theorem.
Afterwards, we show that these proofs can be interpreted as
polynomial-time reductions to the corresponding computational problems.

Let $P \subset\R^d$ be a point set. We say a point $\cc \in \R^d$ has
\emph{Tukey depth} $\tau$ with respect to $P$ if and only if all
closed halfspaces that contain $\cc$ also contain at least $\tau$
points from $P$. The centerpoint theorem guarantees that there always
exist points with large Tukey depth.

\begin{theorem}[{Centerpoint theorem~\cite[Theorem~1]{Rado1946}}]\label{thm:centerpoint}
 Let $P \subset \R^d$ be a point set. Then, there exists a point $\qq \in
 \R^d$ with Tukey depth
 $\tau \geq \left\lceil\frac{|P|}{d+1}\right\rceil$. \qed
\end{theorem}

We call a partition of $P$ into $m$ sets $T_1,\dots,T_m$ a
\emph{Tverberg $m$-partition} if and only if $\bigcap_{i=1}^m
\conv(T_i) \neq \emptyset$. Tverberg's theorem guarantees that there
are always large Tverberg partitions.

\begin{theorem}[Tverberg's theorem~\cite{Tverberg1966}]\label{thm:tverberg}
  Let $P \subset \R^d$ be a point set of size $n$. Then, there always exists a
  Tverberg $\left\lceil\frac{|P|}{d+1}\right\rceil$-partition for $P$.
  Equivalently, let $P$ be of size $(m-1)(d+1)+1$ with $m \in \N$. Then, there
  exists a Tverberg $m$-partition for $P$.
\end{theorem}

Note that Theorem~\ref{thm:tverberg} directly implies Theorem~\ref{thm:centerpoint}.
A point $\cc$ in the intersection of a Tverberg
$\left\lceil\frac{|P|}{d+1}\right\rceil$-partition has Tukey depth at
least $\left\lceil\frac{|P|}{d+1}\right\rceil$ since every halfspace
that contains $\cc$ must contain at least one point from each set in
the Tverberg partition. We present
Sarkaria's proof of Tverberg's theorem~\cite{Sarkaria1992} with
further simplifications by \Barany and Onn~\cite{BaranyOn1997} and
Arocha \etal~\cite{ArochaBaBrFaMo2009}.
The main tool is the following lemma that establishes a notion of
duality between the intersection of convex hulls of low-dimensional point sets
and the embrace of the origin of corresponding high-dimensional
point sets. It was extracted from Sarkaria's proof
by Arocha \etal~\cite{ArochaBaBrFaMo2009}. In the following, we denote with $\otimes$
the tensor product.

In the following, we denote with $\otimes$ the
binary function that maps two points $\pp \in \R^d$, $\qq \in \R^m$ to the point
\[
\pp \otimes \qq =
\begin{pmatrix}
  (\qq)_1 \pp \\
  (\qq)_2 \pp \\
  \vdots \\
  (\qq)_m
  \pp
\end{pmatrix} \in \R^{dm}.
\]
It is easy to verify that $\otimes$ is bilinear, i.e., for all $\pp_1,\pp_2 \in
\R^d$, $\qq \in \R^m$, and $\alpha_1,\alpha_2 \in \R$, we have
\[
\lt(\alpha_1 \pp_1 + \alpha_2 \pp_2\rt) \otimes \qq
= \alpha_1 \lt(\pp_1 \otimes \qq\rt) + \alpha_2 \lt(\pp_2 \otimes \qq\rt)
\]
and similarly, for all $\pp \in \R^d$, $\qq_1,\qq_2 \in \R^m$, and
$\alpha_1,\alpha_2 \in \R$, we have
\[
\pp \otimes \lt(\alpha_1 \qq_1 + \alpha_2 \qq_2\rt)
= \alpha_1 \lt(\pp \otimes \qq_1\rt) + \alpha_2 \lt(\pp \otimes \qq_2\rt).
\]

\begin{lemma}[Sarkaria's lemma~\cite{Sarkaria1992},~{\cite[Lemma~2]{ArochaBaBrFaMo2009}}]
\label{lem:sarkaria}
Let $P_1,\dots,P_m \subset \R^d$ be $m$ point sets and let $\qq_1,\dots,\qq_{m}
\subset \R^{m-1}$ be $m$ vectors with $\qq_i = \e_i$ for $i \in [m-1]$ and
$\qq_m = -\1$. For $i \in [m]$, we define
\[
  \UP{P}_i = \set{\TwoRowVec{\pp}{1} \otimes \qq_i
  \midd \pp \in P_i}\subset \R^{(d+1) (m-1)}.
\]

Then, the intersection of convex hulls $\bigcap_{i=1}^m \convv{P_i}$ is nonempty
if and only if $\;\bigcup_{i=1}^m \UP{P}_i$ embraces the origin.
\end{lemma}
\begin{proof}
  Assume there is a point $\pp^\star \in \bigcap_{i=1}^m \convv{P_i}$.
  For $i \in [m]$ and $\pp \in P_i$, there then exist coefficients
  $\lambda_{i,\pp} \in \Rp$ that sum to $1$ such that
  $\pp^\star = \sum_{\pp \in P_i} \lambda_{i,\pp}$.
  Consider the points $\up{\pp}_i \in \convv{\UP{P}_i}$, $i \in [m]$, that we
  obtain by using the same convex coefficients for the points in $\UP{P}_i$,
  i.e., set
\[
  \up{\pp}_{i}
  = \sum_{\pp \in P_i} \lambda_{i,\pp}
          \lt(\TwoRowVec{\pp}{1} \otimes \qq_i\rt) \in \convv{\UP{P}_i}.
\]
We claim that $\sum_{i=1}^{m} \up{\pp}_i = \0$ and thus
$\0 \in \convv{\bigcup_{i=1}^{m} \UP{P}_i}$. Indeed, we have
\begin{multline*}
  \sum_{i=1}^{m} \up{\pp}_i = \sum_{i=1}^{m} \sum_{\pp \in P_i} \lambda_{i,\pp}
  \left( \TwoRowVec{\pp}{1} \otimes \qq_i \right)
  = \sum_{i=1}^{m} \left( \sum_{\pp \in P_i} \lambda_{i,\pp}
  \TwoRowVec{\pp}{1} \right) \otimes \qq_i
  = \sum_{i=1}^{m} \TwoRowVec{\pp^\star}{1} \otimes \qq_i
  \\
  = \TwoRowVec{\pp^\star}{1} \otimes \left( \sum_{i=1}^{m} \qq_i \right)
  = \TwoRowVec{\pp^\star}{1} \otimes \0 = \0,
\end{multline*}
where we use the fact that $\otimes$ is bilinear.

Assume now that $\bigcup_{i=1}^m \UP{P}_i$ embraces the origin and we want to
show that $\bigcap_{i=1}^m \convv{P_i}$ is nonempty. Then, we can express the
origin as a convex combination $\sum_{i=1}^{m}
\sum_{\up{\pp} \in \UP{P}_i} \lambda_{i,\up{\pp}} \up{\pp}$ with
$\lambda_{i,\up{\pp}} \in\Rp$ for $i \in
[m]$ and $\up{\pp} \in \UP{P}_i$, and $\sum_{i=1}^{m} \sum_{\up{\pp} \in
\UP{P}_i} \lambda_{i,\up{\pp}} = 1$. Hence, we have
\[
\0 = \sum_{i=1}^{m} \sum_{\up{\pp} \in \UP{P}_i} \lambda_{i,\up{\pp}}
\left( \TwoRowVec{\pp}{1} \otimes \qq_i \right)
 = \sum_{i=1}^{m} \left( \sum_{\up{\pp} \in \UP{P}_i} \lambda_{i,\up{\pp}}
 \TwoRowVec{\pp}{1} \right) \otimes \qq_i,
\]
where we use again the fact that $\otimes$ is bilinear.
By the choice of $\qq_1,\dots,\qq_m$, there is (up to multiplication with a scalar)
exactly one linear dependency: $\0 = \sum_{i=1}^{m} \qq_i$.
Thus,
\[
\sum_{\up{\pp} \in \UP{P}_1} \lambda_{1,\up{\pp}} \TwoRowVec{\pp}{1}
= \dots =
\sum_{\up{\pp} \in \UP{P}_m} \lambda_{m,\up{\pp}} \TwoRowVec{\pp}{1} =
\TwoRowVec{\pp^\star}{c},
\]
where $\pp^\star \in \R^d$ and $c \in \R$. In particular, the last equality
implies that
\[
\sum_{\up{\pp} \in \UP{P}_1} \lambda_{1,\up{\pp}}
= \dots =
\sum_{\up{\pp} \in \UP{P}_m} \lambda_{m,\up{\pp}}
= c.
\]
Now, since for all $i \in [m]$ and $\up{\pp} \in \UP{P}_i$, the coefficient
$\lambda_{i,\up{\pp}}$ is nonnegative and since the sum $\sum_{i \in [m]}
\sum_{\up{\pp} \in \UP{P}_i} \lambda_{i,\up{\pp}}$ is $1$, we must have $c =
1/m \in (0,1]$. Hence,
the point $m \pp^\star$ is common to all convex
hulls $\convv{P_1},\dots,\convv{P_m}$.
\end{proof}

Little work is now left to obtain Tverberg's theorem from
Lemma~\ref{lem:sarkaria} and the colorful \Caratheodory theorem.

\begin{proof}[Proof of Theorem~\ref{thm:tverberg}]\label{thm:tverberg:proof}
Let $P = \set{\pp_1,\dots,\pp_n} \subset \R^d$ be a point set of size
$n=(d+1)(m-1)+1$ and let $P_1,\dots,P_m$ denote $m$ copies of $P$.
For each set $P_j \subset \R^d$, $j \in [m]$, we construct a $((d+1)
(m-1))$-dimensional set $\UP{P}_j$ as in Lemma~\ref{lem:sarkaria}, i.e.,
\[
  \UP{P}_j = \set{ \up{\pp}_{i,j} = \pp_i \otimes \qq_j \midd \pp_i \in P}
  \subset \R^{(d+1) (m-1)} = \R^{n-1}.
\]
For $i \in [n]$, we denote with $\UP{C}_i \subseteq \bigcup_{j=1}^m \UP{P}_j$
the set of points $\set{\up{\pp}_{i,j}  \midd j \in [m]}$ that correspond to
$\pp_i \in P$
and we color these points with color $i$. For $i \in [n]$, note that
Lemma~\ref{lem:sarkaria} applied to $m$ copies of the singleton set $\set{\pp_i}
\subseteq P$ guarantees that the color class $\UP{C}_i \in
\R^{n-1}$ embraces the origin. Hence, we have $n$ color classes
$\UP{C}_1,\dots,\UP{C}_n$ that embrace the origin in $\R^{n-1}$. Now,
by Theorem~\ref{thm:colcara}, there is a
colorful choice $\UP{C} = \set{\up{\cc}_1,\dots,\up{\cc}_n} \subseteq
\bigcup_{i=1}^n \UP{C}_i$ with $\up{\cc}_i \in \UP{C}_i$ that embraces the
origin, too. Because $\UP{C}$ embraces the origin,
Lemma~\ref{lem:sarkaria} guarantees that the convex hulls of the sets $T_j = \set{
\pp_i \in P \midd \up{\pp}_{i,j} \in \UP{C}}$, $j \in [m]$, have a point in
common.
Moreover, since all points in $\bigcup_{j=1}^m \UP{P}_j$
that correspond to the same point in $P$ have the same
color, each point $\pp_i \in P$ appears in exactly one set $T_j$, $j \in [m]$.
Thus, $\mc{T} =\set{T_1,\dots, T_j}$ is a Tverberg $m$-partition of $P$.
\end{proof}

Similar to the Tukey depth, the simplicial depth is a further notion of data
depth. Let again be $P \subset \R^d$ be a point set and $\qq \in \R^d$ a point.
Then, the \emph{simplicial depth} $\delta_P(\qq)$ of $\qq$ with respect to $P$
is the number of distinct $d$-simplices that contain $\qq$ with vertices in $P$.
The first selection lemma states that for fixed $d$, there is always a point
with asymptotic optimal simplicial depth.

\begin{theorem}[First selection lemma~{\cite[Theorem~5.1]{Barany1982}}]\label{thm:fslemma}
  Let $P \subset \R^d$ be a set of points and consider $d$ constant. Then, there
  exists a point $\qq \in \R^d$ with $\delta_P(\qq) = \Om{|P|^{d+1}}$.
\end{theorem}

The main argument of \Barany's proof of the first selection lemma is the
following lemma.

\begin{lemma}\label{lem:depths}
  Let $P \subset \R^d$ be a point set and let $\mc{T}$ be a Tverberg
  $m$-partition of $P$, where $m \in \N$. Then any point $\cc \in \bigcap_{T \in
  \mc{T}} \convv{T}$ has simplicial depth $\sigma_P(\cc)$ at least
  $\lt\lceil\frac{m^{d+1}}{(d+1)^{d+1}}\rt\rceil$.
\end{lemma}
\begin{proof}
  Let $T_i$ denote the $i$th element of $\mc{T}$ and color it with color $i$.
  Now by Theorem~\ref{thm:colcara}, there exists for every $(d+1)$-subset $I \subseteq
  [m]$ a colorful choice $C_I$ with respect to the color classes $T_i$, $i \in
  I$, that embraces $\cc$. Furthermore,
  each index set $I$ induces a unique colorful choice $C_I$.
  Thus, there are at least $\binom{m}{d+1} \geq  \frac{m^{d+1}}{(d+1)^{d+1}}$
  distinct $\cc$-embracing $d$-simplices with vertices in $P$.
\end{proof}

The first selection lemma is now an immediate consequence of Lemma~\ref{lem:depths}
and Theorem~\ref{thm:tverberg}.

We define the computational problems that correspond to the
centerpoint theorem, Tverberg's theorem, and the first selection lemma as follows.

\begin{definition}\label{def:polydesc}
We define the following search problems:
 \begin{itemize}
  \item\Centerpoint\hfill
  \begin{description}
    \item[Given] a set $P \subset \Q^d$ of size $n$,
    \item[Find] a centerpoint.
  \end{description}
  \item\Tverberg\hfill
  \begin{description}
    \item[Given] a set $P \subset \Q^d$ of size $n$,
    \item[Find] a Tverberg $\lceil \frac{n}{d+1}\rceil$-partition.
  \end{description}
  \item\SimCenter\hfill
  \begin{description}
    \item[Given] a set $P \subset \Q^d$ of size $n$,
    \item[Find] a point $\qq \in \Q^d$ with $\sigma_P(\qq) \geq f(d)
    n^{d+1}$, where $f: \N \mapsto \R_+$ is an arbitrary but fixed function.
  \end{description}
\end{itemize}
\end{definition}

Finally, interpreting the presented proofs as algorithms, we obtain the
following result.

\begin{lemma}\label{lem:ccppolydesc}
  Given access to an oracle for \CCP, \Tverberg  can be
  solved in $\Oh{n^3}$ time. Furthermore, \Centerpoint and  \SimCenter can be
  solved in $\Oh{L n^3}$ time, where $L$ is the length of the input.
\end{lemma}
\begin{proof}
  As show in the proof of Theorem~\ref{thm:tverberg}, to compute a Tverberg partition,
  it suffices to lift $m=\lt\lceil \frac{n}{d+1}\rt\rceil$ copies of the input point
  set $P \subset \Q^d$ with Lemma~\ref{lem:sarkaria} and then query the oracle for
  \CCP. Lifting
  one point needs $\Oh{dm} = \Oh{n}$ time and hence we need $\Oh{n^3}$ time in
  total. Then, any point in the intersection of the computed Tverberg
  $m$-partition $\mc{T} = \left\{ T_1,\dots,T_m\right\}$ is a solution to
  \Centerpoint and \SimCenter. Using the algorithm
  from~\cite{AnstreicherBo1992}, we can compute a Tverberg point in time
  $\Oh{L n^3}$ by solving the linear program

  \newcommand{\noborder}{\multicolumn{1}{c}{}}
\begin{equation*}
  \lt(\begin{array}{|*{3}{c}|c|*{3}{c}|*{3}{c}}
  \cline{1-3}
  & & & \noborder & & & & -1 & & \\
  & T_1 & & \multicolumn{4}{c|}{0} &  & \ddots & \\
  & &  & \noborder & & & & & & -1 \\
  1 & \dots & 1 & \noborder & & & & 0 & \dots & 0 \\ \cline{1-3}\cline{8-10}
  \noborder& & \noborder & \multicolumn{1}{c}{\ddots} & & & &
  \multicolumn{3}{c}{\vdots} \\
  \cline{5-10}
  \noborder& &\noborder & & & &   & -1 &        & \\
  \multicolumn{4}{c|}{0} & & T_m & &    & \ddots & \\
  \noborder& &\noborder & & &  & &    & & -1 \\
  \noborder& &\noborder & & 1 & \dots & 1 &    0 & \dots & 0 \\
  \cline{5-7}
  \end{array}\rt)
  \xx =
  \lt(\begin{array}{c}
  0\\
  \vdots \\
  0 \\
  1 \\ \hline
  \vdots \\ \hline
  0 \\
  \vdots \\
  0 \\
  1
  \end{array}\rt)
  \text{~s.t.\ } \xx \geq \0,
\end{equation*}
where $L$ is the length of the input.
\end{proof}

\section{Equivalent Instances of the Colorful \Caratheodory Problem in General
Position}
\label{sec:app:eqccp}

The application of Sarkaria's lemma in the reductions to
\CCP creates color classes whose positive span does not have full dimension. To
be able to transfer upper bounds on the complexity of \CCP to its
descendants, we need to be able to deal with degenerate
position. In this chapter, we show how to
ensure general position of \CCP instances by extending known perturbation
techniques for linear programming to our setting. More formally, let $I =
(C_1,\dots,C_d, \bb)$ be a \CCP instance, where $\bb \in  \Q^d \setminus \{\0\}$
and each color
class $C_i \subset \Q^d$, $i \in [d]$, ray-embraces
$\bb$. Then, we want to construct in
polynomial time $d$ sets
$\CA_1,\dots,\CA_d \subset \Z^d$ and a point $\ba \in \Z^d$ that
have the following properties:
\begin{enumerate}[label=(P\arabic{enumi})]
  \item \textbf{Valid instance with integer coordinates:}\label{p:inst}
  The points $\{\ba\} \cup \lt(\bigcup_{i=1}^d \CA_i\rt) \subset \Z^d$ have
  integer coordinates. Furthermore, the point $\ba$ is not the origin and each
  color class $\CA_i$, $i \in [d]$, ray-embraces $\ba$ and has size $d$.
  \item \textbf{$\bb$ avoids linear subspaces:} \label{p:subspace}
  The point $\ba$ is not contained in the linear span of any
  $(d-1)$-subset of $\bigcup_{i=1}^d \CA_i$.
  \item \textbf{Polynomial-time equivalent solutions:} \label{p:eqsol}
  Given a colorful choice $\CA \subseteq\bigcup_{i=1}^d \CA_i$ that
  ray-embraces $\ba$, we can compute in polynomial time a colorful choice $C
  \subseteq \bigcup_{i=1}^d C_i$ that ray-embraces $\bb$.
\end{enumerate}
Note that by~\ref{p:subspace}, if $P \subset
\bigcup_{i=1}^d \CA_i$ ray-embraces $\ba$, then $|P| \geq d$ and thus
$\ba \in \inter \poss{P}$. In particular by~\ref{p:inst}, $\ba$ is contained
in the interior of $\pos(\CA_i)$ for $i \in [d]$.

In the next section, we develop tools to ensure non-degeneracy of linear systems
by a small deterministic perturbation of polynomial bit-complexity.
The approach is similar to already existing perturbation techniques
for linear programming as in~\cite[Section~10-2]{Dantzig1963}
and~\cite{MegiddoCh1989} but extends to a more general setting in which the
matrix is also perturbed. Based on these results, we then show in
Section~\ref{sec:eqccp:const} how to construct
\CCP instances with properties~\ref{p:inst}--\ref{p:eqsol}.

\subsection{Polynomials with Bounded Integer Coefficients}\label{se:eqccp:pert}
In the following, we consider equation systems
\begin{equation}\label{eq:leps}
  L_\eps: A \xx = \bb,
\end{equation}
where $A$ is a $(d \times n)$-matrix with $n \geq d$ and $\bb$ is a $d$-dimensional
vector. Furthermore, the entries of both $A$ and $\bb$ are polynomials in $\eps$
with integer coefficients. For a fixed $\tau \in \R$, we denote with $A(\tau)$
and $\bb(\tau)$ the matrix $A$ and the vector $\bb$ that we obtain by
setting $\eps$ to $\tau$ in $A$ and $\bb$, respectively. Similarly, we denote
with $L_\tau$ the linear system $L_\tau: A(\tau) \xx = \bb(\tau)$.
We show that for any fixed $\tau > 0$ that is sufficiently small in the size
of the coefficients in the polynomials, the linear system $L_\tau$ is
non-degenerate.

For $m \in \N$, we denote with
\[
  \Poly{m} = \left\{ p(\eps) = \sum_{i=0}^{k} \alpha_i \eps^i \midd k \in \N_0,\,
  \text{~and }|\alpha_i| \in [m]_0 \text{ for } i \in [k]_0 \right\}
\]
the set of polynomials with integer coefficients that have absolute value at
most $m$. The following lemma guarantees that no polynomial in $\Poly{m}$ has
a root in a specific interval whose length is inverse proportional to $m$.

\begin{lemma}\label{lem:poly}
  Let $p \in \Poly{m}$ be
  a nontrivial polynomial with $m \in \N$. Then, for all $\eps \in \lt(0,
  \frac{1}{2m}\rt)$, we have $p(\eps) \neq 0$.
\end{lemma}
\begin{proof}
We write $p(\eps) = \sum_{i = 0}^k \alpha_i \eps^i$.
Let $j = \min \{ i \in [k]_0 \mid \alpha_i \neq 0\}$.
Since $p$ is nontrivial, $j$ exists.
Without loss of generality, we
assume $\alpha_{j} > 0$ (otherwise, we multiply $p(\eps)$ by $-1$). For all
$\eps \in \lt(0, \frac{1}{2m}\rt)$, we have
\begin{equation*}
  p(\eps)
  = \sum_{i=0}^{k} \alpha_i \eps^{i}
  \geq \eps^{j} - 2 m \eps^{j+1}
  = \eps^{j} \left( 1 - 2m\eps\right) > 0
\end{equation*}
since $\eps < \frac{1}{2m}$ and hence $p(\eps)\neq 0$ for all $\eps \in \lt(0,
\frac{1}{2m}\rt)$.
\end{proof}
We now use Lemma~\ref{lem:poly} to prove non-degeneracy of
the linear system $L_\eps$ if $\eps$ is fixed but small enough and the degrees of
the monomials in $L_\eps$ are sufficiently separated.
We say $d$ polynomials
$p_1,\dots,p_d \in \Poly{m}$ are
\emph{$(k_1,\dots,k_d)$-separated with gap $g$} if $p_i$ has a nontrivial
monomial of degree $k_i$ and
$p_i$ has no nontrivial monomial of a degree in $\{ k_j - g,\dots, k_j + g \mid j \in
[d] \setminus \{i\}\} \cup \{k_i - g,\dots,k_i-1\}$.

\begin{lemma}\label{stm:perturb}
Let $L_\eps: A \xx = \bb$ be a system of equations as defined in (\ref{eq:leps})
such that
the entries of $A$ and $\bb$ are polynomials in $\Poly{m}$, where $m \in \N$.
Furthermore, suppose that the polynomials in $A$ have degree at most $k_0$ and
$(\bb)_1,\dots,(\bb)_d$ are $(k_1,\dots,k_d)$-separated with gap
$(d-1)k_0$. Set
\[
M = d!(k_0 + 1)^{d-1}(k + 1) m^d,
\]
where $k$ is the maximum degree of $(\bb)_1,\dots,(\bb)_d$. Then, for all $\eps
\in \lt(0, \frac{1}{2M}\rt)$, the linear system $L_\eps$ is non-degenerate.
\end{lemma}
\begin{proof}
We show that for all fixed $\tau \in \lt(0,\frac{1}{2M}\rt)$, the vector $\bb(\tau)$
is not contained in
the linear span of any $d-1$ columns from $A(\tau)$.
We can ensure that $A(\tau)$ has rank $d$ for all fixed
$\tau \geq 0$ by
extending $A$ with the canonical basis of $\R^d$. Then, the entries of
the extended matrix are still polynomials from
$\Poly{m}$ and their degrees are at most $k_0$. Moreover, if for some fixed
$\tau \in \lt(0,\frac{1}{2M}\rt)$, there are $d-1$
columns from the original matrix whose linear span contains $\bb(\tau)$, then
the same holds for the extended matrix.

Let now $\tau \in \lt(0,\frac{1}{2M}\rt)$ be fixed and let $A'$ be a submatrix of $A$
such that $A'(\tau)$ is a basis of $A(\tau)$. Then, the linear system
\[
    L': A'(\tau) \xx = \bb(\tau)
\]
has a unique solution $\xx^\star$. By Cramer's rule, we have
\[
\lt(\xx^\star\rt)_j = \frac{\det A'_j(\tau)}{\det A'(\tau)},
\]
where $j \in [d]$ and $A'_j$ is obtained from the matrix $A'$ by replacing the
$j$th column with $\bb$. Using Laplace expansion, we can express $\det
A'_j$ as
\[
\det A'_j = \sum_{i=1}^{d} (-1)^{i+j} b_i  \det C_{i,j},
\]
where $b_i = (\bb)_i$ and $C_{i,j}$ is the matrix that we obtain by omitting the
$i$th row and the $j$th column from $A'_j$. Next, we apply the Leibniz formula
and write $\det C_{i,j}$ as
\[
  \det C_{i,j} = \sum_{\sigma \in S_{d-1}} \sgn(\sigma)
  \prod_{i=1}^{d-1}(C_{i,j})_{i, \sigma(i)} = c_{i,j}(\eps),
\]
where $c_{i,j}(\eps)$ is a polynomial in $\eps$. Since the polynomials in $A'$
have degree at most $k_0$, the degree of $c_{i,j}$ is at most $(d-1)k_0$.
Because the polynomials in $A'$ have integer coefficients with
absolute value at most $m$, the
coefficients of $c_{i,j}$ are integers, and
the sum of their absolute values can be bounded by
$M' = (d-1)! \big((k_0+1) m\big)^{d-1}$. Hence, $c_{i,j} \in \Poly{M'}$.
Now, since $\det A'(\tau) \neq 0$, at least one of the polynomials
$c_{1,j},\dots,c_{d,j}$, say $c_{i^\star,j}$, is nontrivial.
Let $k'_{i^\star} \leq (d-1)k_0$ be the minimum degree of a nontrivial monomial in $c_{i^\star,j}$.
First, we observe that since $b_{i^\star}$ has a nontrivial monomial of degree
$k_{i^\star}$ and no nontrivial monomial of degree $k_{i^\star} - (d-1)k_0, \dots, 
k_{i^\star} - 1$, 
the polynomial $(-1)^{i^\star+j} b_{i^\star} c_{i^\star,j}$ has a
nontrivial monomial of degree $k' = k_{i^\star} + k'_{i^\star}$. Second, for $i
\in [d]$, $i \neq i^\star$, the polynomial $(-1)^{i+j} b_{i} 
c_{i,j}$ has no monomial of degree $k'$ since $c_{i,j}$ has degree at most
$(d-1)k_0$ and the polynomials $b_1,\dots,b_d$ are $(k_1,\dots,k_d)$-separated
with gap $(d-1)k_0$. Thus, $\det A'_j$ is a nontrivial polynomial. Moreover,
since the polynomials
$b_i$ and $c_{i,j}$ have integer coefficients for $i \in [d]$, so does $\det A'_j$.
Using that the sum of absolute values of the coefficients of $c_{i,j}$ is
bounded by $M'$, we can bound the sum of absolute values of coefficients in $\det A'_j$ by $M = d(k+1)m M'$ and hence
$\det A'_j \in \Poly{M}$, where $k = \max \lt\{\deg b_i \midd i \in [d]\rt\}$.
Then, Lemma~\ref{lem:poly} guarantees that $\det A'_j$
has no root in the interval $\lt(0, \frac{1}{2M}\rt)$. In particular, $\det
A'_j(\tau)
\neq 0$ and hence $(\xx^\star)_j \neq 0$ for all $j \in [d]$.
This means that $\bb(\tau)$ is not contained in
the linear span of any $d-1$ columns from $A(\tau)$.
Since $\tau \in \lt(0, \frac{1}{2M}\rt)$ was arbitrary, the claim follows.
\end{proof}

\subsection{Construction}\label{sec:eqccp:const}
Let $C'_1,\dots,C'_d \subset \Q^d$ be $d$ sets that ray-embrace $\bb' \in
\Q^d$. By applying \Caratheodory's theorem, we can ensure that $|C'_i| \leq d$
for $i  \in [d]$.  First, we rescale the points to the integer grid.  For a
point $\pp' \in \Q_d$, we set $z(\pp') = |\psi| \pp'$, where $\psi \in \Z$ is
the absolute value of the least common multiple of the denominators of
$(\pp')_1,\dots,(\pp')_d$.
Clearly, $z(\pp')$ has integer coordinates and can be represented with a number
of bits polynomial in the number of bits needed for $\pp'$. For $i \in [d]$, let
$C_i = \lt\{ z(\pp') \midd \pp' \in C'_i\rt\}$ be the rescaling of $C'_i$, and
set $\bb = z(\bb')$. Then, the bit complexity of the \CCP instance
$C_1,\dots,C_d, \bb$
is polynomial in the bit-complexity of the original instance.  Moreover, since
$\poss{\pp'} = \poss{z(\pp')}$ for all $\pp' \in \Q^d$, the rescaled color classes
$C_i$, $i \in [d]$, ray-embrace $\bb$ and if a colorful choice $C
\subseteq \bigcup_{i=1}^d C_i$ ray-embraces $\bb$, then the original
points $C' \subset \bigcup_{i=1}^d C'_i$ ray-embrace $\bb'$.  By a
similar rescaling, we can further assume that $\| \bb \|_1 \geq \| \pp \|_1$ for
all $\pp \in \bigcup_{i=1}^d C_i$.

We now sketch how the remaining construction of the equivalent instance
$C_1^\apx,\dots,C_d^\apx, \bb^\apx$ in general position proceeds. First, we
ensure for $i \in [d]$ that $\bb$ lies in the interior of $\poss{C_i}$ by
replacing each point $\pp$ in $C_i$ by a set $P_\eps(\pp)$ of slightly
perturbed points that contain $\pp$ in the interior of their convex hull.
Second, we perturb $\bb$. Lemma~\ref{stm:perturb} then shows that in both steps a
perturbation of polynomial bit-complexity suffices to ensure
properties~\ref{p:subspace} and~\ref{p:eqsol}.

For a point $\pp \in \R^d$, we denote with
\[
P_\eps(\pp) = \left\{ \pp + \eps \e_i, \pp - \eps \e_{i} \mid i \in [d]\right\}
\]
the vertices of the $\ell_1$-sphere around $\pp$ with radius $\eps$.
Let $C_{i}(\eps) = \bigcup_{\pp \in C_i}
P_\eps(\pp)$, $i \in [d]$, denote the $i$th color class in which all points $\pp$ have been
replaced by the corresponding set $P_\eps(\pp)$. Since for
$i \in [d]$, we have $\bb \in \poss{C_i}$ and since each point $\pp \in C_i$ is
contained in the interior of $\poss{P_\eps(\pp)}$, it follows that $\bb \in
\inter \poss{C_i(\eps)}$ for $\eps > 0$. Next, we denote with
\[
\bb(\eps) = \bb + \begin{pmatrix}
  \eps^{d}\\
  \eps^{2d}\\
  \vdots\\
  \eps^{d^2}
\end{pmatrix} \in \R^d
\]
the vector $\bb$ that is perturbed by a vector from the moment curve.
The following lemma shows that for $\eps$ small enough,
Property~\ref{p:subspace} holds for $C_1(\eps), \dots, C_d(\eps)$ and
$\bb(\eps)$.
Let $m$ be the largest absolute value of a coordinate in
$C_1, \dots, C_d, \bb$ and set $N = d!m^d$.

\begin{lemma}\label{stm:p2}
  For all $\eps \in \lt(0,N^{-2}\rt]$,
  there is no $(d-1)$-subset $P \subset \bigcup_{i=1}^d C_i(\eps)$ with
  $\bb(\eps) \in \lspan P$.
\end{lemma}
\begin{proof}
  Let $A$ denote the matrix $\big( C_1(\eps) \dots C_d(\eps)\big)$. Then, there
  exists a subset $P \subset \bigcup_{i=1}^d C_i(\eps)$ with
  $|P|<d$ that contains $\bb(\eps)$ in its linear span if and only if the linear
  system $L_\eps: A \xx = \bb(\eps)$ is degenerate. The polynomials in $A$ all have
  degree at most $1$ and the polynomials $(\bb(\eps))_i$, $i \in [d]$, are
  $\lt(d, 2d, \dots, d^2\rt)$-separated with gap $d-1$. Setting $k_0 = 1$ and
  $k=d^2$ in
  Lemma~\ref{stm:perturb} implies that $L_\eps$
  is non-degenerate for all $\eps \in \lt(0,\frac{1}{2M}\rt)$, where $M = d!
  2^{d-1} (d^2 + 1) m^d$. Assuming that $m\geq 2$ and that $d\geq 4$, we can
  upper bound $2^d$ by $m^d$ and $(d^2+1)$ by $d!$. Hence, we have
  \[
    2M = d! 2^d (d^2 + 1) m^d < \lt(d! m^d\rt)^2 = N^2,
  \]
  and thus the claim follows.
\end{proof}
In the following, we set $\eps_0$ to $N^{-2}$.
Note that Lemma~\ref{stm:p2} holds in particular for $\eps = \eps_0$, and thus a
deterministic perturbation of polynomial bit-complexity suffices. In the next
lemma, we show that the perturbed color classes
still ray-embrace the perturbed $\bb$.

\begin{lemma}\label{stm:p1}
  For $i \in [d]$, the set $C_i(\eps_0)$ ray-embraces $\bb(\eps_0)$.
\end{lemma}
\begin{proof}
Fix some color class $C_i$ and
let $\ve{m}_{\eps_0} = \bb(\eps) - \bb$ be the
perturbation vector for $\bb$. Since
$C_i$ ray-embraces $\bb$, we can express $\bb$
as a positive combination $\sum_{\pp \in C_i} \psi_\pp \pp$,
where $\psi_\pp \geq 0$ for all $\pp \in C_i$. 
Then, 
\[
\bb(\eps_0) = \bb +  \ve{m}_{\eps_0} 
 = \lt(\sum_{\pp \in C_i} \psi_\pp \pp\rt) + \ve{m}_{\eps_0} 
  = \sum_{\pp \in C_i} \psi_\pp \lt(\pp + \frac{1}{s}\ve{m}_{\eps_0} \rt),
\]
where $s = \sum_{\pp \in C_i} \psi_\pp$. We show that $\pp +
\frac{1}{s} \ve{m}_{\eps_0} \in \poss{P_{\eps_0}(\pp)}$ for all $\pp \in C_i$. Since 
$P_{\eps_0}(\pp) \subseteq C_i(\eps_0)$ for all $\pp \in C_i$, this then implies
$\bb(\eps_0) \in \poss{C_i(\eps_0)}$.
First, we claim that $s \geq 1$. Indeed, we have
\[
  \lt\| \bb \rt\|_1 = \lt\| \sum_{\pp \in C_i} \psi_\pp \pp \rt\|_1 \leq
  \sum_{\pp \in C_i} \psi_\pp \lt\| \pp \rt\|_1 \leq
  s\| \bb \|_1,
\]
where the last inequality is due to our assumption $\|
\bb \|_1 \geq \| \pp \|_1$, for $\pp \in C_i$. 
Now,
\[
  \lt\| \frac{1}{s}\ve{m}_{\eps_0} \rt\|_1 < d \eps_0^{d} \leq \eps_0,
\]
for $\eps_0 \leq 1/2$,
and thus $\pp + \frac{1}{s}\ve{m}_{\eps_0}$ lies in the $\ell_1$-sphere around
$\pp$ with radius
$\eps_0$ for all $\pp \in C_i$. By construction of
$P_{\eps_0}(\pp)$, we then have $\pp + \frac{1}{s}\ve{m}_{\eps_0} \in
\convv{P_{\eps_0}(\pp)} \subset \poss{P_{\eps_0}(\pp)}$, as claimed.
\end{proof}

As a consequence of Lemma~\ref{stm:p2}, we can show that colorful choices for the
perturbed instance that ray-embrace $\bb(\eps_0)$, ray-embrace
$\bb$ if the perturbation is removed.

\begin{lemma}\label{stm:p3}
  Let
  $C = \lt\{ \cc_1,\dots,\cc_d \rt\}$ be set
  such that $\cc_i \in C_i(\eps_0)$ for $i \in [d]$ and such that
  $\bb(\eps_0) \in \pos(C)$. Then, the set $C' = \lt\{ \pp \midd i \in [d],\,
  \cc_i \in P_{\eps_0}(\pp)\rt\}$ ray-embraces $\bb$.
\end{lemma}
\begin{proof}
We prove the statement by letting $\eps$ go continuously from $\eps_0$ to $0$.
This corresponds to moving the points in $C$ and $\bb(\eps)$ continuously
from their perturbed positions back to their original positions. We argue that
throughout this motion, $\bb(\eps)$ cannot escape the embrace of the 
colorful choice. 

The coordinates of the points in $C$ are defined by polynomials in the parameter
$\eps$, and we write $C(\eps)$ for the parametrized points.
Then, 
$C = C(\eps_0)$ and $C' = C(0)$.
By Lemma~\ref{stm:p2}, for all $\eps \in (0,\eps_0]$, the point
$\bb(\eps)$ does not lie in any linear subspace spanned by $d-1$
points from $C(\eps)$. It follows that initially $\bb(\eps_0) \in \inter
\poss{C(\eps_0)}$ 
and therefore  $\bb(\eps) \in \inter \poss{C(\eps)}$ for all $\eps \in (0,\eps_0]$.
Assume now that $\bb(0) \notin \poss{C(0)}$. Then, there exists a 
hyperplane $h$ through $\0$ that strictly separates $\bb(0)$ from $C(0)$.
Because the $\ell_2$-distance between $h$ and any point in $C(0) \cup
\{\bb(0)\}$
is positive, there is a $\tau \in (0, \eps_0)$
such that $h$ separates $\bb(\tau)$ from $C(\tau)$, and hence also
from $\poss{C(\tau)}$. 
This is impossible, since we showed that $\bb(\eps) \in \inter\poss{C(\eps)}$
for all $\eps \in (0,\eps_0]$.
\end{proof}

We can now combine the previous lemmas to obtain our desired result on
equivalent instances for \CCP.
\begin{lemma}\label{stm:ccpgpos}
Let $I = \lt(C'_1,\dots,C'_d, \bb'\rt)$ be an instance of \CCP, where $C'_i
\subset \Q^d$
ray-embraces the point $\bb' \in \Q^d$
for all $i \in [d]$. Then, we can construct in polynomial time an instance
$I^\apx=\lt(\CA_1,\dots,\CA_d,\ba\rt)$ of \CCP with
properties~\ref{p:inst}--\ref{p:eqsol}.
\end{lemma}
\begin{proof}
We construct the point sets $C_1(\eps_0), \dots, C_d(\eps_0)$ and the point
$\bb(\eps_0)$ as discussed above.
Since $\log \eps_0^{-1} $ is polynomial in the size of $I$, this needs
polynomial time.
By Lemma~\ref{stm:p1}, each color class $C_i(\eps_0)$ ray-embraces
$\bb(\eps_0)$, so we can apply \Caratheodory's theorem to reduce the size of
$C_i(\eps_0)$ to $d$ while maintaining the property that
$\bb(\eps_0)$ is ray-embraced. Again, we need only polynomial time for
this step.
Finally, as described at the beginning of this section, we rescale the points to
lie on the
integer grid in polynomial time. Let $\CA_i$ denote the resulting point set
for $C_i(\eps_0)$, where $i \in [d]$, and let $\ba$ be the point $\bb(\eps_0)$
scaled to the integer grid. Then, properties~\ref{p:inst}--\ref{p:eqsol} are
direct consequences of this construction and Lemmas~\ref{stm:p2},~\ref{stm:p1},
and~\ref{stm:p3}.
\end{proof}

\section{The Colorful \Caratheodory Theorem is in \PLS}
\label{sec:app:pls}
\subsection{The Complexity Class \PLS}

The complexity class \emph{polynomial-time local search}
(\PLS)~\cite{JohnsonPaYa1988,AartsLe2003,AartsMiKo2007}
captures the complexity of local-search problems that can be solved by a
local-improvement algorithm, where each improvement step can be
carried out in polynomial time, however the number of necessary
improvement
steps until a local optimum is reached may be exponential.
The existence of a local optimum is guaranteed as the
progress of the algorithm can be measured using a potential function
that strictly decreases with each improvement step.

More formally, a problem in \PLS is a relation $\mc{R}$ between a set of
\emph{problem instances} $\mc{I} \subseteq \{0,1\}^\star$ and a set of
\emph{candidate solutions} $\mc{S} \subseteq \{0,1\}^\star$. Assume further the
following.

\begin{itemize}
\item The set $\mc{I}$ is polynomial-time verifiable. Furthermore,
there exists an algorithm that, given an instance $I \in \mc{I}$ and a
candidate solution $s \in \mc{S}$, decides in time $\poly(|I|)$ whether
$s$ is a \emph{valid} candidate solution for $I$.
In the following, we denote with $\mc{S}_\mc{I} \subseteq \mc{S}$
the set of valid candidate solutions for a fixed instance $I$.
\item There exists a polynomial-time algorithm that on input $I \in
\mc{I}$ returns a valid candidate solution $s \in \mc{S}_\mc{I}$.
We call $s$ the \emph{standard solution}.
\item There exists a polynomial-time algorithm that on input $I \in
\mc{I}$ and $s \in \mc{S}_\mc{I}$ returns a set $N_{I,s} \subseteq
\mc{S}_\mc{I}$ of valid candidate solutions for $I$. We call $N_{I,s}$ the
\emph{neighborhood} of $s$.
\item There exists a polynomial-time algorithm that on input $I \in
\mc{I}$ and $s \in \mc{S}_I$ returns a number $c_{I,s} \in \Q$. We
call $c_{I,s}$ the \emph{cost} of $s$.
\end{itemize}

We say a candidate solution $s \in \mc{S}$ is a \emph{local optimum} for an
instance $I \in \mc{I}$ if $s \in \mc{S}_I$ and for all $s' \in N_{I,s}$, we have
$c_{I,s} \leq c_{I,s'}$ in case of a minimization problem, and $c_{I,s} \geq
c_{I,s'}$ in case of a maximization problem.
The relation $\mc{R}$ then consists of all pairs $(I,s)$ such that $s$ is a
local optimum for $I$. This formulation implies a simple algorithm, that we call
the \emph{standard algorithm}: begin with the standard solution, and then
repeatedly invoke the neighborhood-algorithm to improve the current solution
until this is not possible anymore. Although each iteration of this
algorithm can be carried out in polynomial time, the total number of iterations
may be exponential. There are straightforward examples in which this algorithm takes
exponential time and even more, there are \PLS-problems for which it is
\PSPACE-complete to compute the solution that is returned by the standard
algorithm~\cite[Lemma~15]{AartsLe2003}.

Similar to \PPAD, each problem instance $I$ of a \PLS-problem can be seen as a
simple graph searching problem on a graph $G_I=(V,E)$. The set of nodes is the
set of valid candidate solutions for $I$ and there is a directed edge from $u
\in \mc{S}_I$ to $v \in \mc{S}_I$ if $v \in N_{I,u}$ and $c_{I,v} < c_{I,u}$ if
it is a minimization problem, and otherwise if $c_{I,v} > c_{I,u}$. Then, the
set of local optima for $I$ is precisely the set of sinks in $G_I$. Because the
costs induce a topological ordering of the graph, at least one sinks exists.

\subsection{A \PLS Formulation of the Colorful \Caratheodory Problem}

The proof of the colorful \Caratheodory theorem by
\Barany~\cite{Barany1982} admits a straightforward formulation of
\CCP as a \PLS-problem. The only difficulty resides in the computation of the
potential
function: given a set of $d$ points $C \subset \Q^d$ and a point $\bb
\in \Q^d$, we need to be able to compute the point $\pp^\star \in \pos(C)$ with
minimum $\ell_2$-distance to $\bb$ in polynomial time. This problem can be
solved with convex quadratic programming.

We say a matrix $B \in \R^{n\times n}$ is \emph{positive semidefinite} if $B$ is
symmetric and for
all $\xx \in\R^n$, we have $\xx^T B \xx \geq 0$. Then, a \emph{convex quadratic
program} is given by
\begin{align*}
  Q: & \begin{aligned}[t]
  & \min c(\xx) \\
  \text{s.t.} & \begin{aligned}[t]
                 A\xx &=\bb, \\
                   \xx&\geq\0,
                \end{aligned}
  \end{aligned}
\end{align*}
where $\xx \in \R^n$, $\bb \in \Q^d$, $A \in\Q^{d \times n}$, and
the cost function $c : \R^n \mapsto \R$ is defined as
\[
  c(\xx) = \frac{1}{2} \xx^T B \xx + \qq^T \xx,
\]
where the matrix $B \in\Q^{n \times n}$ is positive
semidefinite and $\qq \in \Q^n$. We say a vector $\xx \in \R^n$ is a
\emph{feasible solution} for $Q$ if $A \xx = \bb$ and $\xx \geq \0$.
Furthermore, we say feasible solution $\xx \in \R^n$ is \emph{optimal} for $Q$
if there is no feasible solution $\xx' \in \R^n$ such that $c(\xx') < c(\xx)$.
Convex quadratic programs are known to be solvable in $\Oh{\poly(d,n) L}$ time,
where $L$ is the length of the quadratic program in
binary~\cite{kozlov1980polynomial,KapoorVa1986}.

\begin{lemma}\label{lem:pls_costs}
  Let $C \subset \Q^d$ be a set of size $d$ and let $\bb \in \Q^d$ be a point
  such that $C$ and $\bb$ can be encoded with $L$ bits.
  Then, we can compute the point $\pp^\star \in \pos(C)$ with minimum
  $\ell_2$-distance to $\bb$ in $\Oh{\poly(d) L}$ time.
\end{lemma}
\begin{proof}
  First, we observe that it is sufficient to compute the point $\pp^\star \in
  \pos(C)$ such that 
  \[
  \lt\| \pp^\star - \bb \rt\|_2^2 = \sum_{i=1}^d (\pp^\star - \bb)_i^2
  \]
  is minimum. Let $A$ be the matrix
\begin{equation*}
\newcommand{\diagblock}{\multicolumn{1}{|c}{1} & \multicolumn{1}{c|}{-1}}
  A = \lt(\begin{array}{*{7}{c}|c|c}
  \cline{1-2}
    \diagblock & &  &  & & 0 & & \\ \cline{1-2}\cline{3-4}
    & & \diagblock & & & & C & -\bb  \\ \cline{3-4}
    & & & & \ddots & & & &   \\ \cline{6-7}
    0 & & & & & \diagblock &   \\ \hline
    0 & \multicolumn{6}{c}{\cdots} & 0 & 1
  \end{array}\rt) \in \Q^{(d+1) \times (3d+1)}
\end{equation*}
and let $\bb'$ denote the vector
\[
\bb' = \left(\begin{array}{c}
    0 \\
    \vdots \\
    0 \\ \hline
    1
\end{array}\right) \in \Q^{d+1}.
\]
Furthermore, let $\xx \in \R^{3d+1}$ be a feasible solution to the linear
system
\begin{equation}
\label{eq:pls:linsystem}
A \xx = \bb',\, \xx \geq \0
\end{equation}
and let $\cc_1,\dots,\cc_d$ denote the
points in $C$ ordered according to their respective column indices in
$A$. Write $\xx$ as
\[
\xx = \left(
\begin{array}{*{11}c}
  x_1^+ &
  x_1^- &
  x_2^+ &
  x_2^- &
  \dots &
  x_d^+ &
  x_d^- &
  \multicolumn{1}{|c}{\psi_1} &
  \dots &
  \multicolumn{1}{c|}{\psi_d} &
  x_\bb
\end{array}
\right)^T \in \R^{3d+1},
\]
where $x_i^+,x_i^- \in \R_+$ for $i \in [d]$, $\psi_i \in \R_+$ for $i \in [d]$,
and $x_\bb \in \R_+$. Since $\xx\geq 0$, the point
\[
  \pp = \sum_{i=1}^{d} \psi_i \cc_i
\]
is contained in the positive span of $C$. Furthermore, by the last equality of
(\ref{eq:pls:linsystem}), we have $x_\bb = 1$ and thus for $i \in [d]$, the $i$th
equality of (\ref{eq:pls:linsystem}) is equivalent to
\begin{equation}
  \label{eq:pls:itheq}
  x_i^+ - x_i^- = (\pp)_i - (\bb)_i.
\end{equation}
Now, let $B'$ denote the matrix
\begingroup
\newcommand{\twocells}[2]{\multicolumn{1}{|c}{#1} & \multicolumn{1}{c|}{#2}}
\newcommand{\upperBlock}{\twocells{1}{-1}}
\newcommand{\lowerBlock}{\twocells{-1}{1}}
\[
B' = \left(
\begin{array}{*7c}
\cline{1-2}
\upperBlock & & & & & 0 \\
\lowerBlock & & & & & \\ \cline{1-2} \cline{3-4}
& & \upperBlock & & & \\
& & \lowerBlock & & & \\ \cline{3-4}
& & & &\ddots & & \\  \cline{6-7}
& & & & & \upperBlock \\
0 & & & & & \lowerBlock \\ \cline{6-7}
\end{array}
\right) \in \Q^{(2d) \times (2d)}
\]
\endgroup
and set
\begin{equation*}
  B = \lt(\begin{array}{c|c}
    2 B' & Z_{(2d)\times (d+1)} \\ \hline
    \multicolumn{2}{c}{Z_{(d+1) \times (3d+1)}}
  \end{array}\rt) \in \Q^{(3d+1) \times (3d+1)},
\end{equation*}
where $Z_{a\times b} \in \Q^{a\times b}$ denotes the all-$0$ matrix with $a$
rows and $b$ columns. We claim that $\frac{1}{2} \xx^T B \xx = \| \pp - \bb
\|_2^2$. Indeed, by definition of $B$ and using~(\ref{eq:pls:itheq}), we have
\begin{equation*}
\frac{1}{2} \xx^T B \xx
= \sum_{i=1}^{d} x^+_i \lt(x^+_i - x^-_i\rt) + x^-_i
    \lt(x^-_i - x^+_i\rt)
= \sum_{i=1}^{d} \lt(x^+_i - x^-_i\rt)^2
= \sum_{i=1}^{d} \lt((\pp)_i - (\bb)_i\rt)^2
= \lt\|\pp - \bb\rt\|_2^2.
\end{equation*}
Because $B$ is symmetric, this further implies that $B$ is positive
semidefinite.

Let now $\xx^\star$ be an optimal
solution to the convex quadratic program
\begin{align*}
  \begin{aligned}[t]
  & \min \frac{1}{2} \xx^T B \xx \\
  \text{s.t.} & \begin{aligned}[t]
                 A\xx &=\bb, \\
                   \xx&\geq\0.
                \end{aligned}
  \end{aligned}
\end{align*}
Then, the point
\[
\pp^\star = \sum_{i=2d+1}^{3d} \lt(\xx^\star\rt)_i \cc_i \in \Q^d
\]
is contained in the positive span of $C$.
Moreover, since $\frac{1}{2} (\xx^\star)^T B \xx^\star = \|\pp^\star -
\bb\|_2^2$ is minimum over all feasible solutions and hence over all points in
the positive span of $C$, $\pp^\star$ is the point
in $\pos(C)$ with minimum $\ell_2$-distance to $\bb$. Using the algorithm from
\cite{KapoorVa1986} or \cite{kozlov1980polynomial}, we can compute $\pp^\star$
in $\Oh{\poly(d) L}$ time.
\end{proof}

Having an algorithm to compute the potential function in polynomial time,
we only need to translate the above proof of the colorful \Caratheodory theorem
to the language of \PLS.

\begin{theorem}\label{thm:pls}
  The problems \CCP, \Centerpoint, \Tverberg, and \SimCenter are in $\PPAD \cap \PLS$.
\end{theorem}
\begin{proof}
  By Theorem~\ref{thm:ppad}, \CCP is in \PPAD. We now give a formulation of \CCP as a
  \PLS-problem. Then statement is then implied by Lemma~\ref{lem:ccppolydesc}.

  The set of problem instances $\mc{I}$ consists of all tuples
  $(C_1,\dots,C_d,\bb)$, where $d \in \N$, $\bb \in \Q^d$, $\bb \neq \0$, and
  for all $i \in [d]$, we have $C_i \subset \Q^d$ and $C_i$ ray-embraces $\bb$.
  The set of candidate solutions $\mc{S}$ then consists of all $d$-sets $C
  \subset \Q^d$, where $d \in \N$. Furthermore, for a given instance $I =
  (C_1,\dots,C_d,\bb)$, we
  define the set of valid candidate solutions $\mc{S}_I$ as the set of all
  colorful choices with respect to $C_1,\dots,C_d$. Using linear programming, we
  can check whether a given tuple $I= (C_1,\dots,C_d,\bb)$ is contained in
  $\mc{I}$ and
  clearly, we can check in polynomial time whether a set $C \subset \Q^d$ is a
  colorful choice with respect to $I$ and hence whether $C \in \mc{S}_I$.

  Let now $I \in \mc{I}$ be a fixed instance and $s \in \mc{S}_I$ a valid
  candidate solution. We then define the neighborhood $N_{I,s}$ of $s$ as the
  set of all colorful choices that can be obtained
  by swapping one point in $s$ with another point of the same color. The set
  $N_{I,s}$ can be constructed in polynomial time.

  We define the cost $c_{I,s}$ of a colorful choice $s$ as the minimum
  $\ell_2$-distance of a point in $\pos(s)$ to $\bb$. Using the algorithm from
  Lemma~\ref{lem:pls_costs}, we can compute $c_{I,s}$ in polynomial time. Finally, we
  set the standard solution the colorful choice that consists of the first point
  from each color class.
\end{proof}

\section{The Polytopal Complex}
\label{sec:app:polycompl}

We begin with the following standard lemma that bounds the
bit-complexity of basic feasible solutions for a linear program.

\begin{lemma}\label{lem:lpsol}
Let $L: A \xx = \bb$ be a linear system, where
$A \in \Z^{d\times n}$ and $\bb \in \Z^d$. Furthermore,
let $B$ be a feasible basis for $L$ and let $\xx$ be the corresponding
basic feasible solution. Let $m$ denote the largest absolute value of the
entries in $A$ and $\bb$, and set $N = d!m^d$. Then for $i \in \ind{B}$, we
have $\lt|(\xx)_i\rt| = \frac{n_i}{|\det A_\ind{B}|}$, where $n_i \in [N]_0$ and
$\lt|\det A_\ind{B}\rt| \in [N]$. For $i \in [n] \setminus \ind{B}$, we have $(\xx)_i
= 0$.
\end{lemma}
\begin{proof}
  Set $A' = A_\ind{B}$.
  By definition of a feasible basis, we have $\det A' \neq 0$, and by definition
  of a basic feasible solution $\xx$, we have $A' \xx_\ind{B} = \bb$ with
  $\xx \geq \0$ and $(\xx)_j = 0$ for $j \in [n] \setminus \ind{B}$. Applying
  Cramer's rule~\cite{MatousekGa2007}, we can express the
  $i$th coordinate of $\xx_\ind{B}$ as $\det A'_i/\det
  A'$, where $i \in [d]$ and $A'_i$ is the matrix that we obtain by replacing
  the $i$th column of $A'$ with $\bb$. Using the Leibniz formula, we can bound the
  determinant:
\[
  \lt|\det A'\rt| = \Biggl|\sum_{\sigma \in S_d} \sgn(\sigma)
  \prod_{i=1}^d \lt(A'\rt)_{i, \sigma(i)}\Biggr| \leq d!\, m^d =  N.
\]
And similarly, $\lt|\det A'_i\rt| \leq N$ can be obtained.
Because $\xx$ is a basic feasible solution, we have
\[
\frac{\det A'_i}{\det A'} = (\xx)_i  \geq 0.
\]
Moreover, since $A'$ and $\bb$ contain
only integer entries, the determinants $\det A'$ and $\det A'_i$
are integers. The implies the statement.
\end{proof}

Next, using the techniques from Section~\ref{sec:app:eqccp}, we can show that a
deterministic perturbation of polynomial bit-complexity ensures a non-degenerate
intersection of the parameter regions with $\MM$.

\begin{lemma}\label{g:red}
There exists a constant $c \in \N$ with $c\geq 3$ such that for $\eps = N^{-cd}$
the following holds. Let $B$ be an arbitrary but fixed feasible basis of $\LC$.
Let $h_j \subset \R^d$ denote the hyperplane
\[
  h_j = \left\{ \mm \in \R^d \midd \left(\er_{B,\cc_\mm}\right)_j =
  0 \right\},
\]
and set $H_\Phi = \lt\{ h_j \mid j \in \lt[d^2\rt] \setminus \ind{B} \rt\}$.
Furthermore, let $H_\square$ denote the set of supporting hyperplanes for
the facets of the unit cube in $\R^d$. Then, for all $k$-subsets $H'$ of  $H_\Phi \cup
H_\square$, the intersection $\bigcap_{h \in H'} h$ is either empty or has
dimension $d-k$. In particular, if $k > d$, the intersection must
be empty.
\end{lemma}
\begin{proof}
Let $H'$ be a $k$-subset of $H_\Phi \cup H_\square$, and suppose that
$\bigcap_{h \in H'} h \neq \emptyset$.
We denote with $H'_\Phi = H' \cap H_\Phi$ the hyperplanes from
$H_\Phi$ and similarly, we denote with $H'_\square = H' \cap H_\square$ the
hyperplanes from $H_\square$. Set $R = \lt[d^2\rt] \setminus \ind{B}$ and let
$\phi_1 < \dots < \phi_n \in R$ be the indices such that
$H'_\Phi = \{h_{\phi_1},\dots,h_{\phi_n}\}$, where $n = |H'_\Phi|$.
Then the intersection $\bigcap_{i=1}^n h_{\phi_i}$ is the
solution space to the system of linear equations
\begin{equation}\label{eq:red:hphilp}
 \lt(\lt(\cc_{\mm}\rt)_R - \lt(A^{-1}_\ind{B} A_R \rt)^T
 \lt(\cc_{\mm}\rt)_\ind{B}\rt)_{\rank_R(\phi_i)} = 0 \text{\ for $i \in [n]$},
\end{equation}
where $\rank_R(\phi_i)$ denotes the rank of $\phi_i$ in $R$.
We write $\ind{B} = \{ \beta_1,\dots,\beta_d\}$, with $\beta_1 < \dots <
\beta_d$ and $\ve{a}_i = \lt(A^{-1}_\ind{B} A_R\rt)_{\rank_R(\phi_i)}$, for 
$i \in [n]$. Then,
(\ref{eq:red:hphilp}) is equivalent to
\begin{equation}\label{eq:red:hphieps}
 - dN^2 (\mm)_\col{\phi_i}
 + dN^2 \ve{a}_i^T
 \begin{pmatrix}
   (\mm)_\col{\beta_1} \\
   \vdots\\
   (\mm)_\col{\beta_d}
 \end{pmatrix}
 =
  -1-dN^2-\eps^{\phi_i} +
 \ve{a}_i^T
 \begin{pmatrix}
   1 + dN^2 + \eps^{\beta_1}\\
   \vdots\\
   1 + dN^2 + \eps^{\beta_d}
 \end{pmatrix}
 \text{ for $i \in [n]$},
\end{equation}
where $\col{\phi_i}$ and $\col{\beta_i}$ denote the colors of the
columns with indices $\phi_i$ and $\beta_i$, respectively.
Thus, $(\ref{eq:red:hphieps})$ is of the form
\begin{equation}\label{eq:red:lp}
  A_\Phi \mm = \bb_\Phi,
\end{equation}
where $A_\Phi \in \Q^{n \times d}$ and the polynomials $(\bb_\Phi)_i$, $i \in
[n]$, are
$(\phi_1, \phi_2, \dots, \phi_n)$-separated with gap $0$. 
The entries of
$A_\Phi$ are not necessarily integers due to the occurrence of $A^{-1}_\ind{B}$
in the vectors $\ve{a}_i$. By Lemma~\ref{lem:lpsol}, the fractions in
$A^{-1}_\ind{B}$ all have the same denominator: $\det A_\ind{B} \in \Z$.
We set $A'_\Phi = \lt(\det A_\ind{B}\rt) A_\Phi$ and $\bb'_\Phi = \lt(\det
A_\ind{B}\rt) \bb_\Phi$. Then, the linear system
\begin{equation}\label{eq:red:intlp}
  A'_\Phi \mm = \bb'_\Phi
\end{equation}
is equivalent to~(\ref{eq:red:lp}), where $A'_\Phi \in \Z^{n\times d}$ and
$\lt(\bb'_\Phi\rt)_i$ is a polynomial in $\eps$ with integer coefficients and
a nontrivial monomial of degree $\phi_i$ for $i \in [n]$. Let $m'$ denote the
maximum absolute value of the coefficients of $\eps$-polynomials in
$A'_\Phi$ and $\bb'_\Phi$. Since the absolute value of the entries of $A_R$ is
at most $N$ and since by Lemma~\ref{lem:lpsol} the absolute value of the entries in
$A^{-1}_\ind{B}$ is at most $N$, there exists a constant $c' \in \N$ such that
$m' \leq N^{c'}$ and $c'$ is independent of the choice of $B$.

Set $n' = \lt|H'_\square\rt|$. Since we assume that the hyperplanes in
$H'$ have
a point in common and since $H'_\square \subseteq H'$, the hyperplanes in
$H'_\square$ fix the values of exactly $n'$ coordinates
$(\mm)_j$ to either $0$ or $1$. Let $J$ be the indices of the fixed
coordinates and let $J_i \subseteq J$ be the indices of the $(\mm)_j$ that
are set to $i$ for $i = 0,1$. Combining this with (\ref{eq:red:intlp}), we can
express the intersection of hyperplanes in $H'$ as
\begin{equation}\label{eq:red:finalls}
  \lt(A'_\Phi\rt)_{[d] \setminus J} (\mm)_{[d] \setminus J} = \bb'_\Phi - \sum_{j \in J_1} (A'_\Phi)_j.
\end{equation}
The matrix $\lt(A'_\Phi\rt)_J$ is an $n \times (d-n')$ integer matrix,
whose entries have absolute value at most $N^{c'}$ and
the polynomials $p_i = (\bb'_\Phi - \sum_{j \in J_1} (A'_\Phi)_j)_i$, $i \in
[n]$, are $(\phi_1, \phi_2, \dots, \phi_n)$-separated with gap $0$.
Then, Lemma~\ref{stm:perturb} implies
that for all $\eps \in \lt(0,\frac{1}{2M}\rt)$, the right hand vector of
$(\ref{eq:red:finalls})$ cannot lie in the span of $n-1$ columns of the
left hand matrix, where $M= d! (d^2+1) \big(N^{c'}\big)^d$. Thus, for $c=
\max(3,2c')$, we have $N^{-cd} \in \lt(0,\frac{1}{2M}\rt)$.
Since we know that (\ref{eq:red:finalls}) has a solution,
it follows that the rank of (\ref{eq:red:finalls})
must be $n$ and thus the intersection $\bigcap_{h \in H'} h$ has dimension
$d-n-n'=d-k$.
\end{proof}

Note that since $c$ is a constant, the number of bits needed to represent $\eps$
is polynomial in the size of the \CCP instance.  We continue by showing that
the elements from $Q$ are indeed polytopes and by characterizing precisely their
dimension and their facets.

\begin{lemma}\label{lem:para_region}
Let $q = \Phi(f) \cap g \neq \emptyset$ be an element from
$\QQ$, where $f \in \FF$ and $g$ is a face of $\MM$.
Then, $q$ is a simple polytope of dimension $\dim g - \dim f$.
Moreover, if $\dim q > 0$, the set of facets of $q$ can be written as
\[
  \left.\Big\{ \Phi\left(f\right) \cap \facet{g} \neq \emptyset \midd
  \text{$\facet{g}$ is a facet of $g$} \Big\}\right. \cup
  \left\{ \Phi\left(\up{f}\right) \cap g \neq \emptyset \midd \text{$f$ is a
  facet of $\up{f} \in \FF$}\right\}.
\]
\end{lemma}
\begin{proof}
Let $B$ be a feasible basis for a vertex of $f$. As discussed above, the
solution space to the linear system $\LR_{B,f}$ is $\Phi(f)$.
We denote with $H^=_{\Phi(f)}$ the set of hyperplanes that are given by the equality
constraints
\[
  (\er_{B, \mm})_j = 0, \text{ for } j \in \supp{f} \setminus \ind{B},
\]
and we denote with $H^-_{\Phi(f)}$ the set of halfspaces that are given by the $d^2 -
(d + \dim f)$ inequalities
\[
  (\er_{B, \mm})_j \leq 0, \text{ for } j \in \lt[d^2\rt] \setminus \supp{f}
\]
in $\LR_{B,f}$.

Because $g$ is a face of $\MM$ and hence of the unit cube,
we can write it as the intersection of a set $H^=_g$ of $d - \dim g$ hyperplanes
and a set of halfspaces $H^-_g$, where $H^=_g$ and the boundary hyperplanes from
the halfspaces in $H^-_g$ are supporting hyperplanes of facets of the unit cube.

We set $H^= = H^=_g \cup H^=_{\Phi(f)}$ and $H^- = H^-_g \cup H^-_{\Phi(f)}$.
Now, $q$ is the intersection of the affine space
$S^= = \bigcap_{h \in H^=} h$ with the polyhedron $S^- = \bigcap_{h^-
\in H^-}h^-$. Hence, $q$ is a polyhedron and moreover, as $q
\subseteq \MM$, it is a polytope. By Lemma~\ref{g:red}, the
hyperplanes in $H^=$ and the boundary hyperplanes of $H^-$
are in general position, so $q$ is simple.

We now prove $\dim q  = \dim g  - \dim f$.
Because $|H^=_{g}| = d-\dim g$, $|H^=_{{\Phi(f)}}| = \dim f$, and by
Lemma~\ref{g:red}, we have $H^=_g \cap H^=_{{\Phi(f)}} =\emptyset$, the set $H^=$
contains $d-\dim g+\dim f $ hyperplanes. Again by Lemma~\ref{g:red}, the hyperplanes 
from $H^=$ are in
general position, and therefore $\dim S^=  = \max(\dim g - \dim f, -1)$, where we
set $\dim \emptyset = -1$.
Since we assume that $q \neq \emptyset$, it follows that
$\dim S^= \geq 0$, so in particular $\dim f \leq \dim g$. We show
that the dimension does not decrease by intersecting $S^=$ with the
halfspaces in $H^-$.
Fix an arbitrary ordering $h^-_1,\dots,h^-_{m}$, $m = |H^-|$,
of the halfspaces in $H^-$. For $j = 0, 1,\dots, m$, let $\Psi_j$
denote the polyhedron that we obtain by intersecting $S^=$ with the
first $j$ halfspaces $h^-_1,\dots,h^-_j$ from $H^-$. In particular, we
have $\Psi_0 = S^=$ and $\Psi_{m} = q$. Assume for the sake of
contradiction that $\dim q  < \dim S^= $, and let $j^\star$ be such
that $\dim \Psi_{j^\star-1}  = \dim S^= $ and $\dim \Psi_{j^\star}  =
d_{j^\star} < \dim S^= $. There are three possibilities:
(i) $\Psi_{j^\star-1} \cap h^-_{j^\star} = \emptyset$; (ii) $h^-_{j^\star}$
intersects the relative interior of $\Psi_{j^\star-1}$; or
(iii) $h^-_{j^\star}$
intersects only the boundary of $\Psi_{j^\star-1}$.
Now, since $q \neq \emptyset$, Case (i)
is impossible. Since by our assumption, $d_{j^\star} < \dim \Psi_{j^\star
-1}$,
Case (ii) also cannot occur. Hence, $\Psi_{j^\star}$ is a
proper face of $\Psi_{j^\star-1}$.
Then, $\Psi_{j^\star}$ is contained in the  intersection of the
$d-\dim g+\dim f $ hyperplanes from $H^=$ with at least
$\dim S^= -d_{j^\star} = \dim g-\dim f -d_{j^\star}$ boundary hyperplanes of
$h^-_1,\dots,h^-_{j^\star-1}$, and with the boundary hyperplane of 
$h^-_{j^\star}$.  Thus, the $d_{j^\star}$-dimensional polyhedron $\Psi_{j^\star}$
lies in the intersection of at least $d-d_{j^\star}+1$ hyperplanes
from $H^=$ and bounding hyperplanes from $H^-$. Hence, the
hyperplanes from $H^=$ together with the bounding hyperplanes from
$H^-$ are not in general position, a contradiction to Lemma~\ref{g:red}.

We now prove the second part of the statement. Let $\facet{q}$
be a facet of $q$. Since $\dim q > 0$, the facet $\facet{q}$ is nontrivial. Then,
$\facet{q}$ is the intersection of $q$ with a
hyperplane $h^\star$ that is a boundary hyperplane of some halfspace in
$H^-$. Let $h^-$ be the halfspace that generates $h^\star$. If $h^- \in
H^-_g$, then $\facet{g} = g \cap h$ is a facet of $g$ and we have
$\facet{q} = \Phi(f) \cap \facet{g}$. Assume now $h^- \in H^-_{\Phi(f)}$ and let $h$ be
defined by the equation $(\er_{B, \cc_\mm})_j = 0$ for some $j \in \supp{f}
\setminus \ind{B}$. Let $\up{f} \subseteq \PC$ be the face that is defined by
the columns from $A$ with indices $\supp{f} \cup \{j\}$, and note that $f$ is a
facet of $\up{f}$. Then, we can write $\facet{q}$ as
\[
\facet{q} = h^\star \cap q
=  h^\star \cap
  \left(\bigcap_{h \in H^=_{\Phi(f)}} h
        \cap \bigcap_{h^- \in H^-_{\Phi(f)}} h^-
        \cap g
  \right)
= \left(h^\star
        \cap \bigcap_{h \in H^=_{\Phi(f)}} h
        \cap \bigcap_{h^- \in H^-_{\Phi(f)}} h^-
  \right)
  \cap g
\]
and thus $\facet{q}$ contains all parameter vectors in $g$ for which
$\up{f}$ is optimal.

Now, let $\facet{g}$ be a facet of $g$ with $\facet{q} = \Phi(f) \cap \facet{g}
\neq \emptyset$. Then, there exists a boundary
hyperplane $h^\star$ from a halfspace in $H^-_g$ such that $\facet{q} =
h^\star \cap \lt(\bigcap_{h \in H^=} h\rt) \cap \lt(\bigcap_{h^- \in H^-}
h^-\rt)$. Clearly,
$\facet{q}$ is a face of $q$. Furthermore, since
$\facet{q} \neq \emptyset$ the first part of
the lemma implies
\[
\dim \facet{q}  = \dim \facet{g}  - \dim f  = (\dim g - 1)  - \dim f = \dim q  -
1 \geq 0.
\]
Hence $\facet{q}$ is a facet of $q$.
Let now $\up{f} \in \FF$ be a face that has $f$ as a facet with
$\facet{q} = \Phi(\up{f}) \cap g \neq \emptyset$. Then there exists a
boundary hyperplane $h^\star$ of a halfspace in $H^-_{\Phi(f)}$ such that
$\facet{q} =  h^\star \cap \lt(\bigcap_{h \in H^=} h\rt) \cap \lt(\bigcap_{h^-
\in H^-}
h^-\rt)$. As before, $\facet{q}$ is a face of $q$ and since
$\facet{q} \neq
\emptyset$, we get
\[
\dim \facet{q}  = \dim g  - \dim \up{f} = \dim g  - (\dim f + 1) = \dim q  -1
\geq 0.
\]
Thus, $\facet{q}$ is a facet of $q$.
\end{proof}

In particular, Lemma~\ref{lem:para_region} implies that within each $k$-face of
$\MM$, the set of parameter vectors that are optimal for some vertex $v \in \FF$ is
either empty or a $k$-dimensional polytope and the set of parameter vectors that
are optimal for a $k$-face $f \in \FF$ is either empty or a
single point. Furthermore, Lemma~\ref{lem:para_region} immediately bounds the maximum
dimensions of faces in $\FF$.

The next lemma shows that the intersection of any two polytopes in $\QQ$ is
again an element in $\QQ$.

\begin{lemma}\label{lem:para_intersection}
  Let $q_1 = \Phi(f_1) \cap g_1 \in \QQ$ and $q_2 = \Phi(f_2)
  \cap g_2 \in \QQ$ be two polytopes with $q_1\cap q_2 \neq \emptyset$, where
  $f_1,\, f_2 \in \FF$ and
  $g_1,\, g_2$ are faces of $\MM$. Then,
  \[
      q_1 \cap q_2 = \Phi\left(\up{f}\right) \cap \down{g},
  \]
  where $\up{f} \in \FF$ is the smallest face of $\PC$ that contains $f_1$ and
  $f_2$, and $\down{g} = g_1 \cap g_2$.
\end{lemma}
\begin{proof}
We begin with showing that $\Phi(f_1) \cap
\Phi(f_2) = \Phi\left(\up{f}\right)$. Let $\mm \in \Phi(f_1) \cap \Phi(f_2)$ be a
vector. Since $\up{f}$ is the smallest face of $\PC$ that contains $f_1$ and
$f_2$, the face $\up{f}$ is optimal for $\LC_\mm$ and thus $\Phi(f_1) \cap
\Phi(f_2) \subseteq \Phi\left(\up{f}\right)$.
Let now $\mm$ be a parameter vector from $\Phi\left(\up{f}\right)$. Since $f_1$
and $f_2$ are subfaces of
$\up{f}$, the faces $f_1$ and $f_2$ are optimal for $\mm$ and thus we
have $\mm \in \Phi(f_1) \cap \Phi(f_2)$.
Hence, $\Phi\left(\up{f}\right) = \Phi(f_1) \cap
\Phi(f_2)$.
Then, we can express $q_1 \cap q_2$ as
\[
q_1 \cap q_2 = \left(\Phi(f_1) \cap g_1\right) \cap
\left(\Phi(f_2) \cap g_2\right)
= \Phi\left(\up{f}\right) \cap \down{g},
\]
where $\down{g} = g_1 \cap g_2$. Moreover, since $q_1 \cap q_2 \neq \emptyset$
and $\down{g}$ is a face of $\MM$, the face $\up{f}$ is contained in $\FF$.
\end{proof}

Equipped with Lemmas~\ref{lem:para_region} and~\ref{lem:para_intersection}, we
are now ready to show that $\QQ$ is a polytopal complex.

\begin{lemma}\label{stm:para_polycompl}
The set $\QQ$ is a $(d-1)$-dimensional polytopal complex that
decomposes $\MM$.
\end{lemma}
\begin{proof}
Lemma~\ref{lem:para_region} guarantees that every element $q \in
\QQ$ is a polytope in $\R^d$ of dimension at most $d-1$. By the second
part of Lemma~\ref{lem:para_region}, if $\dim q > 0$, all facets of $q$ and hence
inductively all nonempty faces of $q$ are contained in $\QQ$.
Furthermore, since $\emptyset$ is a face of $\MM$, it is contained
in $\QQ$ as well.

Now, let $q_1 ,\, q_2 \in \QQ$ be two polytopes. If $q_1 \cap q_2 =
\emptyset$, then clearly $q_1\cap q_2$ is a face of both polytopes $q_1$ and
$q_2$, so assume $q_1 \cap q_2 \neq \emptyset$. By definition of
$\QQ$, there are faces $f_1,\, f_2 \in
\FF$ and faces $g_1,\, g_2$ of $\MM$ such that $q_1 =
\Phi(f_1) \cap g_1$ and $q_2 = \Phi(f_2) \cap g_2$. Then, we can apply
Lemma~\ref{lem:para_intersection} to express the intersection
of $q_1$ and $q_2$ as $\Phi\left(\up{f}\right) \cap \down{g}$. Since $\up{f} \in
\FF$ and since $\down{g}$ is a face of $\MM$, $q_1 \cap q_2 \in \QQ$.
Moreover, as $\up{f}$ is a superface of $f_1$ and $\down{g}$ is a face of $g_1$,
a repeated application of Lemma~\ref{lem:para_region} shows that $q_1 \cap q_2$ is
a face of $q_1$. Similarly, because $\up{f}$ is a superface of $f_2$ and
$\down{g}$ is a face of $g_2$, a repeated application of
Lemma~\ref{lem:para_region} proves that $q_1 \cap q_2$ is a face of $q_2$, as desired.
\end{proof}

A further implication of Lemmas~\ref{lem:para_region}
and~\ref{lem:para_intersection} is that each polytope in $\QQ$ can be
represented uniquely as the intersection of a
parameter region of a face of $\PC$ and a face of $\MM$.
\begin{lemma}\label{lem:para_unique}
  Let $q \in \QQ$ be a polytope.
  Then, there exists unique pair of faces
  $f,\,  g$, where $f \in \FF$ and $g$ is a face of $\MM$, such that
  $q = \Phi(f) \cap g$.
\end{lemma}
\begin{proof}
Let
$f_1,\, f_2$ be two faces of $\PC$ and let $g_1,\, g_2$ be two faces of $\MM$
such that
\[
q = \Phi(f_1) \cap g_1 = \Phi(f_2) \cap g_2.
\]
Then, by
Lemma~\ref{lem:para_intersection}, we can write $q$ as $\Phi\left(\up{f}\right) \cap \down{g}$, where
$\up{f} \in\FF$ is the smallest face in $\PC$ that contains $f_1$ and
$f_2$ and $\down{g}$ is a face of $g_1$ and of $g_2$. If $\up{f} \neq f_1$ or
$\down{g} \neq g_1$, then by Lemma~\ref{lem:para_region}, 
\[
\dim q = \dim \down{g}  - \dim \up{f} < \dim g_1 - \dim f_1 = \dim q,
\]
a contradiction. Hence,
we must have $\up{f} = f_1$ and $\down{g} = g_1$. Similarly, we must have
$\up{f} = f_2$ and $\down{g} = g_2$, and thus $f_1 = f_2$ and $g_1 = g_2$.
\end{proof}

Lemmas~\ref{stm:face_dim} and~\ref{stm:polycompl} are now immediate consequences
from Lemmas~\ref{lem:para_intersection},~\ref{lem:para_unique},
and~\ref{lem:para_intersection}.

We conclude with the proof of Lemma~\ref{lem:colors}. For this, we need the following
observation that is a direct consequence of Property~(P2) of the \CCP instance.

\begin{observation}\label{g:basis}
  For any feasible
  basis $B$ of $\LC$, the coordinates for $B$ in the corresponding
  basic feasible solution are strictly positive. Equivalently, $\PC$
  is simple.
\end{observation}

\begin{proof}[Proof of Lemma~\ref{lem:colors}]
Let $\xx^\star$
be the basic feasible solution for $B^\star$ with respect to $\LC_\mm$. For the
sake of contradiction, suppose that $B^\star$ contains some vector of $C_{i^\times}$, and
let $k$ be the index of the corresponding coordinate in $\xx^\star$.
By Observation~\ref{g:basis} and Lemma~\ref{lem:lpsol}, we have $\lt(\xx^\star\rt)_k \geq 1/N$.
Hence,
\[
\cc_{\mm}^T \xx^\star \geq \lt(\cc_\mm\rt)_k \lt(\xx^\star\rt)_k \geq
  \left( 1 + d N^2 \right)  \lt(\xx^\star\rt)_k \geq dN + \frac{1}{N}
\]
since $\cc_{\mm} \geq \1$ and $\xx^\star \geq \0$.
By construction, there is a color $i^\star \in [d]$ such that 
$(\cc_{\mm})_j = 1 + \eps^j$
for all columns $j$ with color $i^\star$. Let $\xx^{(i^\star)}$ be the basic
feasible solution for the basis $C_{i^\star}$. By Lemma~\ref{lem:lpsol},
$\left(\xx^{(i^\star)}\right)_j$ is upper bounded by $N$ for all $j \in
\ind{C_{i^\star}}$, so we can lower bound the costs of $\xx^{(i^\star)}$ as follows:
\[
  \cc_{\mm}^T \xx^{(i^\star)} = \sum_{j \in \ind{C_{i^\star}}} (\cc_{\mm})_j
  \left(\xx^{(i^\star)}\right)_j
\leq \sum_{j \in \ind{C_{i^\star}}} \lt(1 + \frac{1}{N^3}\rt)
\left(\xx^{(i^\star)}\right)_j \leq d N + \frac{d}{N^2} < d N + \frac{1}{N},
\]
where we use that $0 < \eps \leq N^{-3}$.
This contradicts the optimality of $B^\star$.
\end{proof}

\section{The Barycentric Subdivison -- Omitted Proofs}
\label{sec:app:barycentric}

\begin{proof}[Proof of Lemma~\ref{lem:fullycolored_colbasis}]
Let $q_0 \subset \dots \subset q_{d-1}$ be the chain that corresponds to
$\sigma$ in $\sd \QQ_\DD$. By Lemma~\ref{stm:unique}, we can write each polytope $q_i
\in \QQ_\DD$ uniquely as $\Phi_\DD(f_i) \cap g_i$, where $i \in [d-1]_0$, $f_i
\in \FF$, and $g_i \in \SS$. By the definition of the barycentric subdivision
and since $\QQ_\DD$ is a $(d-1)$-dimensional polytopal complex,
$q_{i-1}$ is a facet of $q_i$ for $i \in [d-1]$. Then, Lemma~\ref{stm:face_dim}
states that either $g_{i-1}$ is a facet of $g_i$ or $f_{i}$ is a facet of
$f_{i-1}$ for $i \in [d-1]$. Because $\sigma$ is fully-labeled, we must have
$f_i \neq f_j$ for all $i,\, j \in [d-1]_0$ with $i\neq
j$. Hence, $f_{i}$ is a facet of $f_{i-1}$ for $i \in [d-1]$ and thus $g_0 =
\dots = g_{d-1}$. Since $\dim q_{d-1} = d-1$, Lemma~\ref{stm:face_dim}
implies that $\dim f_i = d - 1 - i$ and hence $\lt|\supp{f_i}\rt| = 2d - 1 -i$ for $i
\in [d-1]_0$. In particular, $\dim f_{d-1} = 0$ and thus the columns from
$A_{\supp{f_{d-1}}}$ are a feasible basis for $\LC$.
For $i \in [d-1]$, let $a_{i-1} \in \lt[d^2\rt]$ denote the column index such that
$\supp{f_{i-1}} = \supp{f_i} \cup \{a_{i-1}\}$. Since the faces
$f_0,\dots,f_{d-1}$ have pairwise distinct labels and since $\lt|\supp{f_{i-1}}\rt| =
\lt|\supp{f_i}\rt| + 1$
for $i \in [d-1]$, the column vectors $A_{a_{0}},\dots,A_{a_{d-2}}$ have
pairwise distinct colors by the definition of $\lambda$ (see
(\ref{eq:ppad:labeling})).
Now assume for the sake of contradiction that the columns from
$A_{\supp{f_{d-1}}}$ are not a colorful feasible basis. Then, there is some color
$i^\times \in [d]$ that does not appear in $A_{\supp{f_{d-1}}}$ and hence there
is some color $i^\star \in [d]$ with $\lt|\ind{C_{i^\star}} \cap \supp{f_{d-1}}\rt|
\geq 2$.
Since there is at most one column with color $i^\times$ among
$A_{a_0},\dots,A_{a_{d-2}}$, we have $\lt|\supp{f_i} \cap \ind{C_{i^\times}}\rt|
\leq 1$ for all $i \in [d-1]_0$. Since $\supp{f_i} \supseteq \supp{f_{d-1}}$ for
$i \in [d-1]_0$ and since $\lt|\ind{C_{i^\star}} \cap \supp{f_{d-1}}\rt| \geq
2$, we have
$\lambda(f_i) \neq i^\times$ for all $i \in [d-1]_0$,
a contradiction to $\sigma$ being fully-labeled.
\end{proof}

\begin{proof}[Proof of Lemma~\ref{lem:enc_bij}]
  We begin by showing that the encoding $\en{\sigma}$ of a simplex $\sigma \in
  \FS_k$ is a valid $k$-tuple. Let $q_0 \subset \dots \subset q_{k-1}$ be the
  corresponding
  face chain in $\QQ_\DD$ such that the $i$th vertex of $\sigma$ is the
  barycenter of $q_i \in \QQ_\DD$ and $q_i \neq \emptyset$ for $i \in [k-1]_0$.
  By Lemma~\ref{stm:unique}, for
  each $q_i$, $i \in[k-1]_0$, there exists a unique pair of faces $f_i \in \FF$
  and $g_i \in \SS$ such that $q_i = \Phi_\DD(f_i) \cap g_i$. Because $q_{k-1}
  \neq \emptyset$, we have $\MM(q_{k-1}) = \Phi(f_i) \cap
  g\lt(I^{(k-1)}_0,I^{(k-1)}_1\rt) \neq \emptyset$.
  We further observe that
  $g_i \subset \DD_{[k]}$. Otherwise we would have $q_i = \Phi_\DD(f_i) \cap
  \lt(g_i \cap \DD_{[k]}\rt)$ with $g_i \cap \DD_{[k]} \in \SS$, a contradiction
  to $g_i, f_i$ being the unique pair. Since $q_i \subset \DD_{[k]}$ for $i \in
  [k-1]_0$ and since $\dim \DD_{[k]} = k-1$, we must have $\dim q_i = i$ for $i \in
  [k-1]$. Then,
  Lemma~\ref{stm:face_dim} implies that $\dim g_{k-1} = k-1$ and $\dim f_{k-1} = 0$.
  In particular, $\supp{f_{k-1}}$ is the index set of a feasible basis and
  $\lt|I^{(k-1)}_0 \cup I^{(k-1)}_1\rt| = d- k +1$. Because $g_{k-1} \subset
  \DD_{[k]}$, we have $[d] \setminus [k] \subseteq I^{(k-1)}_0$ and since
  $g_{k-1}$ is the projection of a face of $\MM$, the set $I^{(k-1)}_1$ is
  nonempty. Thus, $I^{(k-1)}_0 = [d] \setminus [k]$ and $\lt|I^{(k-1)}_1\rt|=1$.

  Let now $i \in[k-1]$ be a fixed index and write $\en{q_{i-1}} =
  \lt(\supp{f_{i-1}}, I^{(i-1)}_0, I^{(i-1)}_1\rt)$ and $\en{q_{i}} =
  \lt(\supp{f_i}, I^{(i)}_0, I^{(i)}_1\rt)$.
  Since $q_{i-1}$ is a facet of $q_i$, Lemma~\ref{stm:face_dim} implies that either
  (a) $f_{i}$ is a facet of $f_{i-1}$ and $g_{i-1} = g_i$ or (b) $f_{i-1} = f_i$
  and $g_{i-1}$ is a facet of $g_i$.  In Case~(a), we have $\supp{f_{i-1}} =
  \supp{f_i} \cup \left\{  a_{i-1} \right\}$ and $I^{(i-1)}_0 = I^{(i)}_0$ as
  well as $I^{(i-1)}_1 = I^{(i)}_1$, where $a_{i-1} \in\lt[d^2\rt] \setminus
  \supp{f_i}$. In Case~(b), we have $\supp{f_{i-1}} = \supp{f_i}$. Furthermore,
  since $\MM(g_{i-1})$ is a facet of $\MM(g_i)$, we either have $I^{(i-1)}_0 =
  I^{(i)}_0 \cup \left\{ j_{i-1} \right\}$ and 
  $I^{(i-1)}_1 = I^{(i)}_1$,
  or $I^{(i-1)}_1 = I^{(i)}_1 \cup \left\{ j_{i-1} \right\}$ and
  $I^{(i-1)}_0 = I^{(i)}_0$, for an index
  $j_{i-1} \in [d] \setminus \lt(I^{(i)}_0 \cup I^{(i)}_1\rt)$.
  Thus, $\en{\sigma}$ is a valid $k$-tuple.

  We now show that $\enOperator$ is a bijection. Let $\sigma_1,\sigma_2 \in
  \FS_k$ be two simplices. Since the barycenters of the polytopes in a polytopal
  complex are pairwise distinct, the face chains in $\QQ_\DD$ that corresponds
  to $\sigma_1$ and $\sigma_2$ must differ in at least one face. Then,
  (\ref{eq:enc:simplex}) together with Lemma~\ref{stm:unique} directly implies that
  $\en{\sigma_1} \neq \en{\sigma_2}$.

  Let now $T=\lt(Q_0,\dots,Q_{k-1}\rt)$, $k \in [d-1]$, be a valid $k$-tuple,
  where $Q_i = \lt(S^{(i)}, I^{(i)}_0, I^{(i)}_1\rt)$. For $i \in [k-1]_0$,
  let $g'_i = g\lt(I^{(i)}_0 \cup I^{(i)}_1\rt)$ be the subset of $\MM$ that
  is defined by the index sets $I^{(i)}_0, I^{(i)}_1$. Since $[d] \setminus[k]
  \subseteq I^{(i)}_0$ for all $i \in [k-1]_0$, the projection $g_i=\DD(g'_i)$
  is a subset of $\DD_{[k]}$. Moreover, since $I^{(i)}_1 \neq \emptyset$ for $i
  \in [k-1]_0$, the set $g'_i$ is a face of $\MM$ and hence $g_i \in \SS$.
  Furthermore, since the columns in
  $A_{S^{(k-1)}}$ are a feasible basis, they define a vertex $f_{k-1}$. Because
  $S^{(k-1)} \subseteq S_i$ for $i \in [k-1]_0$, the index set $S_i$ is the
  support of a face $f_i \in \FF$. Set $q_i = \Phi_\DD(f_i) \cap g_i \in \QQ$
  for $i \in [k-1]_0$. Because $g_i \subset \DD_{[k]}$, the polytope $q_i$ is
  also contained in $\DD_{[k]}$.  By Property~\ref{valid:ini} of a valid
  sequence, the intersection
  $\Phi(f_{k-1}) \cap g'_{k-1}$ is nonempty and hence its projection $q_{k-1}$
  onto $\DD$ is nonempty. Then, Lemma~\ref{stm:face_dim} states
  that $\dim q_{k-1} = k-1$. Moreover by Lemma~\ref{stm:face_dim} and
  properties~\ref{valid:facet:f} and~\ref{valid:facet:g} of $T$,
  either $g_{i-1}$ is a facet of $g_i$ or $f_i$ is a facet of $f_{i-1}$ for $i
  \in [k-1]$. Thus by Lemma~\ref{stm:face_dim}, $q_{i-1}$ is a facet of $q_i$, $i \in
  [k-1]$. Then, $\dim q_i = i$ for all $i \in [k-1]_0$ and hence the face chain
  $q_0 \subset \dots \subset q_{k-1}$ defines a $(k-1)$-simplex $\sigma \in
  \FS_k$ with $\en{\sigma} = T$.
\end{proof}

\begin{proof}[Proof of Lemma~\ref{lem:enc_verify}]
  Clearly, we can check if $T$ fulfills all syntactic requirements on valid
  $k$-tuples in polynomial time. Furthermore, we can check in polynomial time
  whether the columns $B$ from $A_{S^{(k-1)}}$ are a feasible basis for a vertex
  $f$. Finally, we express
$
\Phi(f) \cap g\lt(I^{(k-1)}_0, I^{(k-1)}_1\rt)
$
as the solution
  space to the linear system $\LC_{B,f}$ extended by the constraints $\mm \in
  g\lt(I^{(k-1)}_0, I^{(k-1)}_1\rt)$. Then, we can check in polynomial time
  whether this system has a solution.
\end{proof}

The key for Lemma~\ref{lem:enc_algs} is the following lemma that guarantees that
simplices with facets in common have a similar encoding.

\begin{lemma}\label{lem:enc_share}
  Let $\sigma, \sigma' \in \FS_k$ be two simplices, where $k \in [d]$. Then,
  $\sigma$ and $\sigma'$ share a facet if and only if the tuples $\en{\sigma}$
  and $\en{\sigma'}$ agree in all but one position.
  Furthermore, let $\sigma \in \FS_{k}$ and $\up{\sigma} \in \FS_{k+1}$ be two
  simplices, where $k \in [d-1]_0$. Write $\en{\sigma}$ as
  \[
  \en{\sigma} = \lt(Q_0,\dots, Q_{k-1}=\lt(S^{(k-1)},I^{(k-1)}_0,
  I^{(k-1)}_1\rt)\rt).
  \]
   Then, $\sigma$ is a facet of $\up{\sigma}$ if and only if
  \[
  \en{\up{\sigma}} = \lt(Q_0,\dots,Q_{k-1}, \lt(S^{(k-1)},I^{(k-1)}_0
  \setminus \left\{ k+1 \right\}, I^{(k-1)}_1\rt)\rt).
  \]
\end{lemma}
\begin{proof}
  Let $\sigma, \sigma' \in \FS_k$ be two simplices  and let $q_0\subset \dots
  \subset q_{k-1}$ and $q'_0\subset \dots \subset q'_{k-1}$ be the corresponding
  face chains in $\QQ_\DD$. Then $\sigma$ and $\sigma'$ share a facet if and
  only if the face chains agree on all but one position and hence if and only
  if $\en{\sigma}$ and $\en{\sigma'}$ agree on all but one position.

  Let now $\sigma \in \FS_{k}$ and $\up{\sigma} \in \FS_{k+1}$ be two
  simplices. Let $q_0 \subset \dots \subset q_{k-1}$ be the face chain in
  $\QQ_\DD$ that corresponds to $\sigma$ with $\dim q_i = i$ for $i \in
  [k-1]_0$. Similarly, let $\up{q}_0 \subset
  \dots \subset \up{q}_k$ be the face chain in $\QQ_\DD$ that corresponds to
  $\up{\sigma}$ with $\dim \up{q}_i = i$ for $i \in [k]_0$. Furthermore, we write
  $\en{q_{k-1}} = \lt(S^{(k-1)},I^{(k-1)}_0,I^{(k-1)}_1\rt)$ and
  $\en{\up{q}_k} = \lt(S^{(k)},I^{(k)}_0,I^{(k)}_1\rt)$.
  Then, $\sigma$ is a facet of $\up{\sigma}$ if and only if the
  faces $q_0,\dots,q_{k-1}$ appear in the face chain of $\up{\sigma}$ and hence
  if and only if $q_i = q'_i$ for $i \in [k-1]_0$. Moreover, since by
  Lemma~\ref{lem:enc_bij} the encodings $\en{\sigma}$ and $\en{\up{\sigma}}$ are
  valid tuples, the columns of $A_{S^{(k-1)}}$ and $A_{S^{(k)}}$ are feasible
  bases. Since $S^{(k-1)} \subseteq S^{(k)}$ by Property~\ref{valid:facet} of
  valid tuples, we must have $S^{(k-1)} = S^{(k)}$. Moreover, by
  Property~\ref{valid:ini}, we have $I^{(k-1)}_0 = [d] \setminus [k]$,
  $I^{(k)}_0 = [d] \setminus [k+1]$, and $\lt|I^{(k-1)}_1\rt| =
  \lt|I^{(k)}_1\rt| = 1$.
  Because of Property~\ref{valid:facet}, the index set $I^{(k-1)}_1$ is a subset
  of $I^{(k)}_1$ and hence $I^{(k-1)}_1 = I^{(k)}_1$. We conclude that
  \[
  \en{\up{\sigma}} = \lt(\en{q_0},\dots,\en{q_{k-1}}, \lt(S^{(k-1)},
  I^{(k-1)}_0 \setminus \left\{ k+1\right\}, I^{(k-1)}_1\rt)\rt),
  \]
  as claimed.
\end{proof}

\begin{proof}[Proof of Lemma~\ref{lem:enc_algs}]
We begin with the first problem. By Lemma~\ref{lem:enc_share}, if there is a simplex
$\sigma' \in \FS_k$ that shares the facet $\conv\left\{ \vv_j \midd j \in
[k-1]_0,\, j \neq i \right\}$ with $\sigma$, the encodings $\en{\sigma}$ and
$\en{\sigma'}$ agree on all but one position. Thus, there are only polynomially
many possibilities for the encoding of $\en{\sigma'}$ that we can check in
polynomial time with the algorithm from Lemma~\ref{lem:enc_verify}.
Furthermore, Lemma~\ref{lem:enc_share} directly implies polynomial-time algorithms
for the second and third problem.
\end{proof}

\section{The PPAD Graph}
\label{sec:app:ppadgraph}
We begin by characterizing by showing that the graph consists only of paths and
cycles and by characterizing the degree one nodes.

\begin{proof}[Proof of Lemma~\ref{stm:deg}]
Let $\en{\sigma} \in V_k$ be the encoding of a simplex $\sigma \in \FS_k$.
If $\sigma \in \FS_1$ then $\deg \en{\sigma} = 1$ since the only adjacent
node is the encoding of the simplex in $\FS_2$ with $\sigma$ as a facet.
Similarly, if $\en{\sigma} \in V_d$ with $\lambda(\sigma) = [d]$, then $\deg
\en{\sigma}=1$ since the only adjacent node is either the encoding of the single
$[d-1]$-labeled facet of $\sigma$ or the encoding of the simplex in $\FS_d$ that
shares this facet.

If $k>1$ and $\sigma$ has two $[k-1]$-labeled facets, then $\deg \en{\sigma}=2$
since each $[k-1]$-labeled facet is either shared with another simplex in
$\FS_k$ or the facet is itself in $\FS_{k-1}$. Otherwise, if $k<d$ and
$\lambda(\sigma) = [k]$, then we have again $\deg \en{\sigma} = 2$ as there
exists exactly one simplex in $\FS_{k+1}$ with $\sigma$ as a facet and either
the single $[k-1]$-labeled facet of $\sigma$ is shared with another simplex in
$\FS_k$ or it is itself a simplex in $\FS_{k-1}$. Note that actually
Lemma~\ref{lem:enc_bij} implies in this case that the $[k-1]$-labeled facet must
be shared with another simplex in $\FS_k$.
\end{proof}

We continue with the orientation of the edges in $G$. In the following,
we assume that given a node $\en{\sigma} \in V$, we are able to
compute in polynomial time the vertices of the corresponding simplex
$\sigma \in \FS$. We show afterwards how to implement this step.
With this assumption, the orientation can be defined similarly as
in~\cite{Papadimitriou1994}.

Let $\en{\sigma}, \en{\sigma'} \in V_d$ be two adjacent nodes.
By definition, the encoded simplices $\sigma = \conv(\vv_0,\dots,\vv_{d-1})$ and
$\sigma'$ share a
facet $\facet{\sigma}=\conv(\vv_1,\dots,\vv_{d-1})$ with
$\lambda(\facet{\sigma})=[d-1]$. Let the indices be such that $\lambda(\vv_i) =
i$ for $i \in [d-1]$. Then, the edge between $\en{\sigma}$ and $\en{\sigma'}$
is directed from $\en{\sigma}$ to
$\en{\sigma'}$ if and only if the function $\dir(\sigma,\sigma')$
is positive, where
\[
 \dir(\sigma, \sigma') = \sgn
  \det \begin{pmatrix}
    1 & 1 & \dots & 1\\
    \vv_0 & \vv_1 & \dots & \vv_{d-1}
  \end{pmatrix}.
\]
Only for the sake of orientation,  we define a set of
$d-1$ vertices $\wv_2,\dots, \wv_{d}$ with colors $2,\dots,d$ to
lift lower-dimensional simplices in order to avoid dealing with simplices of
different dimensions. For $i = 2,\dots,d$, let $\wv_i \in \R^d$ denote the
parameter vector
\[
  \left( \wv \right)_j =
  \begin{cases}
    2 & \text{if } j< i,\\
    1-2(i-1) & \text{if } j=i,\text{ and}\\
    0 & \text{otherwise,}
  \end{cases}
\]
where $j \in [d]$. Furthermore, we set $\lambda(\wv_i)=i$. Since
$(\wv_i)_i < 0$ for $i=2,\dots,d$, we have $\wv_i \notin \DD$ and for $k < i$,
$\wv_i \notin
\aff(\DD_{[k]})$. However, a quick calculation shows that $\wv_i \in
\aff(\DD_{[i]})$ and that
within $\aff(\DD_{[i]})$, the hyperplane $\aff(\DD_{[i-1]})$ separates
$\e_{i}$ and $\wv_i$. Now, let
$\sigma=\conv(\vv_0,\dots,\vv_{k-1})$ denote a simplex that corresponds
to some node in $G$, where $k \in [d-1]_0$.
Then, we denote with
$\sigma_\wv=\conv(\vv_0,\dots,\vv_{k-1},\wv_{k+1},\dots,\wv_{d})$ the
$(d-1)$-simplex that we obtain by lifting $\sigma$ with our additional
vertices outside of $\DD$. Note that $\sigma_{\wv}$ is non-degenerate by our
choice of $\wv_2,\dots,\wv_d$.
If $\sigma$ is already a $(d-1)$-simplex, we set $\sigma_\wv =
\sigma$. Let now $\en{\sigma}$ and
$\en{\sigma'} \in V$ be two adjacent nodes.
Then the two lifted simplices $\sigma_\wv$ and $\sigma'_\wv$ share a
$[d-1]$-labeled facet. Now, we set $\dir(\sigma, \sigma') = \dir(\sigma_\wv,
\sigma'_\wv)$ and we direct the edge between $\en{\sigma'}$ and
$\en{\sigma}$ as discussed before. The following lemma guarantees
that the orientation of the edge is
the same if seen from either $\sigma$ or $\sigma'$ and that the
only sinks and sources remain the nodes of degree $1$ that are
characterized by Lemma~\ref{stm:deg}.

\begin{lemma}\label{lem:deg_orient}
  The orientation of $G$ is well-defined. Furthermore,
  $\en{\sigma} \in V$ is a sink or a source if and only if
  $\deg \en{\sigma} = 1$ in the underlying undirected graph.
\end{lemma}
\begin{proof}
  Let $\en{\sigma}, \en{\sigma'} \in V$ be two adjacent nodes.
  Assume first that $\en{\sigma}, \en{\sigma'}
  \in V_k$ for some $k \in [d]$. Let $\sigma =
  \conv(\vv_0,\vv_1,\dots,\vv_{k-1})$ and $\sigma'=\conv(\vv'_0,
  \vv_1,\dots,\vv_{k-1})$ denote the encoded simplices with
  $\lambda(\vv_i) = i$ for $i \in [k-1]$.
  That is, $\sigma$ and $\sigma'$ share the facet
  $\facet{\sigma}=\conv(\vv_1,\dots,\vv_{k-1})$. Because both simplices are
  contained in $\FS_k$,
  the two vertices $\vv_0$ and $\vv'_0$ are separated within the
  $(k-1)$-dimensional affine space $\aff(\DD_{[k]})$ by the $(k-2)$-dimensional
  affine space $\aff(\facet{\sigma})$. Since $\wv_{k+1},\dots,\wv_d \notin
  \aff(\DD_{[k]})$, the two vertices $\vv_0$ and $\vv'_0$ are separated in
  $\R^d$ by the
  hyperplane $\aff(\vv_1,\dots, \vv_{k-1},\wv_{k+1},\dots,\wv_d)$. Then, we have
  $\dir(\sigma, \sigma') = -\dir(\sigma', \sigma)$, since
\begin{align*}
  \dir(\sigma, \sigma')
   &= \dir(\sigma_{\wv}, \sigma'_{\wv})\\
   &=
  \sgn \det \begin{pmatrix}
    1 & 1 &  \dots & 1 & 1 & \dots & 1\\
    \vv_0 & \vv_1 & \dots & \vv_{k-1} & \wv_{k+1} & \dots &  \wv_d
  \end{pmatrix}\\
  &=
  - \sgn \det \begin{pmatrix}
    1 & 1 & \dots & 1 & 1 & \dots & 1\\
    \vv'_0 & \vv_1 & \dots & \vv_{k-1} & \wv_{k+1} & \dots &  \wv_d
  \end{pmatrix}\\
   &= -\dir(\sigma'_{\wv}, \sigma_{\wv})\\
  &= -\dir(\sigma', \sigma).
\end{align*}
Let now $\en{\sigma} \in V_{k-1}$ and $\en{\up{\sigma}}
\in V_k$ be two adjacent nodes for some $k \in [d]$. By
definition of $E$, we then have $\lambda(\sigma) = [k-1]$ and $\sigma$ is a
facet of $\up{\sigma}$. We write
$\sigma=\conv(\vv_1,\dots,\vv_{k-1})$ and
$\up{\sigma}=\conv(\vv_0,\vv_1, \dots,\vv_{k-1})$, where the indices are such
that $\lambda(\vv_i) = i$ for $i \in [k-1]$.
Then,
\[
\sigma_\wv = \conv(\vv_1,\dots,\vv_{k-1},\wv_k,\dots,\wv_{d}) \text{~and }
\up{\sigma}_\wv=\conv(\vv_0,\vv_1,\dots,\vv_{k-1},\wv_{k+1},\dots,\wv_{d}).
\]

Hence, $\sigma_\wv$ and $\up{\sigma}_\wv$ share the facet $\facet{\sigma}_\wv =
\conv(\vv_1,\dots,\vv_{k-1},\wv_{k+1},\dots,\wv_{d})$.
By construction,
both vertices $\vv_{0}$ and $\wv_{k}$ are contained in
$\aff(\DD_{[k]})$. Within the $(k-1)$-dimensional affine space $\aff(\DD_{[k]})$,
the vertex $\wv_{k}$ is separated from $\DD_{[k]}$ by the $(k-2)$-dimensional
affine space $\aff(\DD_{[k-1]})$ and hence it is separated from $\vv_0$ by
$\aff(\DD_{[k-1]})$. Since $\sigma \in \FS_{k-1}$, $\sigma$ is a $(k-2)$ simplex
that is contained in $\DD_{[k-1]}$ and thus $\aff(\sigma) = \aff(\DD_{[k-1]})$
separates
$\vv_0$ and $\wv_k$ in $\aff(\DD_{[k]})$.
Now, because $\wv_{k+1},\dots,\wv_k \notin \aff(\DD_{[k]})$,
$\vv_0$ and $\ww_k$ are separated in $\R^d$ by the hyperplane
$\aff(\facet{\sigma}_\wv)$. Again we have $\dir(\sigma, \up{\sigma}) =
-\dir(\up{\sigma}, \sigma)$, since
\begin{align*}
  \dir(\sigma, \up{\sigma})
   &= \dir(\sigma_{\wv}, \up{\sigma}_{\wv})\\
   &=
  \sgn \det \begin{pmatrix}
    1 & 1 &  \dots & 1 & 1 & \dots & 1\\
    \wv_{k} & \vv_1 & \dots & \vv_{k-1} & \wv_{k+1} & \dots &  \wv_{d}
  \end{pmatrix}\\
  &=
  - \sgn \det \begin{pmatrix}
    1 & 1 & \dots & 1 & 1 & \dots & 1\\
    \vv_0 & \vv_1 & \dots & \vv_{k-1} & \wv_{k+1} & \dots &  \wv_{d}
  \end{pmatrix}\\
   &= -\dir(\up{\sigma}_{\wv}, \sigma_{\wv})\\
  &= -\dir(\up{\sigma}, \sigma).
\end{align*}
It remains to show the second part of the statement. Let $\en{\sigma}
\in V$ be a node with two adjacent nodes $\en{\sigma'},
\en{\sigma''}$.
We want to show that the two incident edges are
oriented differently. In any case, the lifted simplices
$\sigma_\wv$ and $\sigma_\wv'$ share a $[d-1]$-labeled facet 
$\facet{\sigma}_{\wv}'$ and similarly, $\sigma_\wv$ and 
$\sigma''_\wv$ share a $[d-1]$-labeled facet
$\facet{\sigma}_{\wv}''$. The facets $\facet{\sigma}_{\wv}'$
and $\facet{\sigma}_{\wv}''$ of $\sigma_\wv$ differ in exactly
one vertex with the same label. Thus, 
the determinants in $\dir(\sigma, \sigma')$ and
$\dir(\sigma, \sigma'')$ differ by exactly one column-swap.
The properties of the determinant now ensure that
$\dir(\sigma, \sigma') = - \dir(\sigma, \sigma'')$, as desired.
\end{proof}

Our next lemma shows that for purposes of orientation, we can
replace the barycenters by arbitrary interior points in the
corresponding parameter faces.
\begin{lemma}\label{lem:proxy}
Let $q_0,\dots,q_{k-1} \subset \R^{d}$ be $k$ polytopes such that $q_0 \subset
\dots \subset q_{k-1}$ and $\dim q_i = i$ for $i \in [k-1]_0$.
Furthermore let $\vv_i$ denote the barycenter of $q_i$ for $i \in [k-1]_0$ and
let $\vv'_0,\dots,\vv'_{k-1}$ be $k-1$ vectors such that
$\vv'_i \in q_i$ and $\aff(\vv'_0,\dots,\vv'_i) = \aff(q_i)$ for all $i \in
[k-1]_0$. Then,
\[
  \sgn \det \begin{pmatrix}
    1 &  \dots & 1 & 1 & \dots & 1\\
    \vv_{0}   & \dots & \vv_{k-1} & \xx_{k+1} & \dots &  \xx_d
  \end{pmatrix}
  =
  \sgn \det \begin{pmatrix}
    1 & \dots & 1 & 1 & \dots & 1\\
    \vv'_0  & \dots & \vv'_{k-1} & \xx_{k+1} & \dots &  \xx_d
  \end{pmatrix},
\]
where $\xx_i \in \R^d \setminus \aff q_{k-1}$, $i \in [d]\setminus [k]$, is an
arbitrary point.
\end{lemma}

\begin{proof}
The prove involves only basic linear algebra, however it is included for
completeness.  We show by induction on $i$ that $\aff(q_i) =
\aff(\vv'_0,\dots,\vv'_i)$ and
that for all $j \in [i]_0$, $\vv'_{j} = \sum_{l=0}^{j} \alpha_{j,l} \vv_l$ is an
affine combination of $\vv_0,\dots,\vv_{j}$ with $\alpha_{j,j} > 0$.

For $i = 0$ the induction hypothesis trivially holds since $\dim q_0 = 0$ and
hence $q_0 = \vv_0 = \vv'_0$. Assume now that $i>0$ and that the induction
hypothesis holds for all $i' < i$. Since $q_{i-1}$ is a facet of $q_i$,
within the $i$-dimensional affine space $\aff(q_i)$,
$q_i$ lies on one
side of the $(i-1)$-dimensional affine space $\aff(q_{i-1})$ and thus it lies
on one side of $\aff(\vv'_0,\dots,\vv'_{i-1})$. Since both $\vv_i$ and $\vv'_i$ lie
on the same side of $\aff(\vv'_0,\dots,\vv'_{i-1})$ in $\aff(q_i)$,
we can write $\vv'_i$ as
$\sum_{l=0}^{i-1} \beta_l \vv'_l + \alpha_i \vv_i$ with $\alpha_i > 0$. By our
induction hypothesis, $\vv'_0,\dots,\vv'_{i-1} \in \aff(\vv_0,\dots,\vv_{i-1})$
and hence the hypothesis holds for $i$.
The claim now follows directly from the properties of the determinant:

\begin{align*}
  &\sgn \det \begin{pmatrix}
    1 &  \dots & 1&\dots& 1 & 1 & \dots & 1\\
    \vv_{0}' & \dots & \vv'_{i} & \dots & \vv_{k-1}' & \xx_{k+1} &
    \dots &  \xx_d
  \end{pmatrix} \\
  = &
  \sgn \det \begin{pmatrix}
    1 &  \dots & 1&\dots& 1 & 1 & \dots & 1\\
    \vv_0 & \dots & \sum_{l=0}^{i} \alpha_{i,l} \vv_l  & \dots &
    \sum_{l=0}^{k-1} \alpha_{k-1,l} \vv_l  & \xx_{k+1} &
    \dots &  \xx_d
  \end{pmatrix} \\
  = &
  \sgn \det \begin{pmatrix}
    1 &  \dots & 1&\dots& 1 & 1 & \dots & 1\\
    \vv_0 & \dots & \alpha_{i,i} \vv_l  & \dots &
    \alpha_{k-1,k-1} \vv_{k-1}  & \xx_{k+1} &
    \dots &  \xx_d
  \end{pmatrix} \\
  = &
  \sgn \det \begin{pmatrix}
    1 &  \dots & 1&\dots& 1 & 1 & \dots & 1\\
    \vv_0 & \dots & \vv_i  & \dots & \vv_{k-1}  & \xx_{k+1} &
    \dots &  \xx_d
  \end{pmatrix},
\end{align*}
where the last equality holds since $\alpha_{i,i} > 0$ for $i \in [k-1]$.
\end{proof}

As the next lemma shows, computing parameter vectors in the relative interior of
faces in $\QQ_\DD$ is computationally feasible.

\begin{lemma}\label{lem:relint_vertices}
Let $\en{\sigma}=\lt(\en{q_0},\dots,\en{q_{k-1}}\rt) \in V$ be a node of $G$,
where $k \in [d]$. Then, we can compute in polynomial time $k-1$ parameter
vectors $\vv_0,\dots,\vv_{k-1}$ such that $\vv_i \in q_i$ and
$\aff(\vv_0,\dots,\vv_i)=\aff(q_i)$ for $i \in [k-1]_0$.
\end{lemma}
\begin{proof}
By definition of the encoding, $q_0$ is a vertex and hence we can choose $\vv_0 =
q_0$. The algorithm iteratively computes now incident edges $e_i =
\conv(\vv_0,\vv_i)$ to $\vv_0$ for $i \in [k-1]$ such that $e_i$ is an edge of
$q_i$ and no
edge of $q_{i-1}$. The resulting vectors have the desired properties:
$\vv_i \in q_i$ and $\aff(\vv_0,\dots,\vv_i)=\aff(q_i)$ for $i \in [k-1]_0$.

We construct these edges as follows. Write $\en{q}_i = \lt(\supp{f_i},
I^{(i)}_0, I^{(i)}_1\rt)$ and let $g_i$ be the
face $g\lt(I^{(i)}_0, I^{(i)}_1\rt)$ of $\MM$ that is encoded by the index sets
$I^{(i)}_0$ and $I^{(i)}_1$.
Since $\en{\sigma}$ is a valid $k$-tuple, the columns $B$ from
$A_{\supp{f_{k-1}}}$ are a feasible basis and moreover, since $\supp{f_{k-1}} \subseteq
\supp{f_i}$ for $i \in [k-1]_0$, the set $B$ is a feasible basis for all faces
$f_i$, $i \in [k-1]_0$. Similar to the proof of Lemma~\ref{lem:enc_verify}, we can
express each polytope $\MM(q_i)$ as the solution to the linear system $\LR_{B,f_i}$
extended by the constraints $\mm \in g_i$, where $i \in [k-1]_0$. Let $L_i$
denote the resulting linear system. Again by the properties of a valid
$k$-tuple, either $\supp{f_{i-1}} = \supp{f_i} \cup \left\{ a_{i-1} \right\}$, where
$a_i \in \lt[d^2\rt] \setminus \supp{f_i}$. Or there is an index $j_{i-1} \in [d]
\setminus \lt(I^{(i)}_0\cup I^{(i)}_1\rt)$ such that $I^{(i-1)}_0 =
I^{(i)}_0\cup\left\{ j_{i-1} \right\}$ and $I^{(i-1)}_1 =
I^{(i)}_1$, or $I^{(i-1)}_0 = I^{(i)}_0$ and $I^{(i-1)}_1 =
I^{(i)}_1\cup\left\{ j_{i-1} \right\}$. This means, that the linear system
$L_{i-1}$ equals the linear system $L_i$ where one inequality becomes tight.
In the following we call this inequality $e_{i}$. Note that $L_0$ is then the
linear system $L_{k-1}$ in which all inequalities $e_1,\dots,e_{k-1}$ are tight.

Assume now that we already have computed the vectors $\vv_0,\dots,\vv_{i-1}$
such that $\vv_j \in q_j$ and $\aff(\vv_0,\dots,\vv_j) =\aff(q_j)$ for $j \in
[i-1]_0$
and we want to compute $\vv_i$, where $i \in [k-1]$. We consider the linear
system $L'_{i}$ that we obtain by relaxing the tight inequality $e_i$ in $L_0$.
Since the solution space of $L_0$ is the vertex $\vv_0$, the solution space to
$L'_{i}$ is an edge $\conv(\vv_0,\vv_i)$. We can compute the other endpoint
$\vv_i$ of this edge in polynomial time by computing the line that is defined by
the equalities in $L'_{i}$ and intersect this
iteratively with the halfspaces that are defined by the inequalities in $L'_i$
while keeping track of the endpoints. Now, we have $\vv_i \in q_i$ since the
solution space of the linear system $L'_i$ is a subset of the solution space of
the linear system $L_i$. Moreover, since in $L_{i-1}$ the inequality $e_i$
is tight, $\vv_i \in q_i \setminus q_{i-1}$ and thus $\aff(\vv_0,\dots,\vv_i)
=\aff(q_i)$.
\end{proof}

The following lemma is now an immediate consequence of
Lemmas~\ref{lem:proxy} and~\ref{lem:relint_vertices}.

\begin{lemma}\label{stm:orientation}
  Let $\en{\sigma}, \en{\sigma} \in V$ be two adjacent nodes. Then, we can
  compute $\dir(\sigma, \sigma')$ in polynomial time. \qed
\end{lemma}
\section{A Polynomial-Time Case}
\label{sec:app:halfhalf}

In the following, we use the same notation as in Section~\ref{sec:ppad} (see
Table~\ref{tab:notation} on Page~\pageref{tab:notation} for an overview). Let $C_1,C_2 \subset \Q^d$ be two color
classes, each of size $d$, let $\bb
\in \Q^d$, $\bb \neq \0$, be a point that is ray-embraced by $C_1$ and by $C_2$,
and let $k \in [d-1]$ be a number.
Although not needed in the algorithm, to comply with the
formulations of our results in Section~\ref{sec:app:eqccp} and
Section~\ref{sec:ppad}, we introduce $d-2$ ``dummy'' color classes $C_3,\dots,C_d$ that
trivially ray-embrace $\bb$ by setting $C_3 = \dots = C_d = \{\bb\}$.
Let $(C'_1,\dots,C'_d,\bb')$ be the instance of \CCP in general position that we
obtain by applying Lemma~\ref{stm:ccpgpos} to $(C_1,\dots,C_d,\bb)$. Then, let $\PC
\subset \Q^{d^2}$ denote the polyhedron that is defined by the linear system
$\LC$ (see (\ref{eq:lp}) on Page~\pageref{eq:lp}) for the instance
$(C'_1,\dots,C'_d,\bb')$. Furthermore, let $\DD_1 = \DD \cap
\conv\lt(\e_1,\e_2\rt)$ denote the edge of the standard simplex $\DD^{d-1}$ that
connects $\e_1$ with $\e_2$ and set $\QQ_{\DD_1} = \lt\{ q \in \QQ_\DD \midd q
\subseteq \DD_1 \rt\}$. Note that by Lemma~\ref{stm:polycompl}, the set
$Q_{\DD_1}$ is a $1$-dimensional polytopal complex that decomposes $\DD_1$.
We begin with the following basic lemma on $\QQ_{\DD_1}$.

\begin{lemma}\label{lem:edges}
  Let $e, e' \in \QQ_{\DD_1}$, $ e\neq e'$, be two adjacent edges with $e =
  \Phi_\DD(f) \cap
  g$ and $e' = \Phi_\DD(f') \cap g'$, where $f, f' \in \FF$ and $g,g' \in \SS$.
  Then, $f$ and $f'$ are vertices of $\PC$ with $\supp{f}, \supp{f'} \subseteq
  \ind{C'_1\cup C'_2}$
  and $\supp{f}, \supp{f'}$ differ in at most one column index.
\end{lemma}
\begin{proof}
By Lemma~\ref{stm:face_dim}, the faces $f,f'$ are vertices of $\PC$. Furthermore,
since $\MM(e), \MM(e') \subset \lspan(\e_1,\e_2)$, Lemma~\ref{lem:colors} implies
that $\supp{f}, \supp{f'} \subseteq \ind{C'_1 \cup C'_2}$. Now, since $e$ and
$e'$ are adjacent, they share a vertex $\vv = \Phi_\DD(f_\vv) \cap g_\vv \in
\QQ_{\DD_1}$, where $f_\vv \in \FF$ and $g_\vv \in \SS$. Then, by
Lemma~\ref{stm:face_dim}, either $f$ is a facet of $f_\vv$ and $g = g_\vv$, or $f =
f_\vv$ and $g_\vv$ is a facet of $g$. Similarly, either $f'$ is a facet of
$f_\vv$ and $g' = g_\vv$, or $f' = f_\vv$ and $g_\vv$ is a facet of $g'$.
Then, Observation~\ref{g:basis} implies the statement.
\end{proof}

Using Lemma~\ref{lem:edges}, we now present
a polynomial-time checkable criterion whether an interval $[\mm_1,\mm_2] \subset
\DD_1$ intersects
an edge $e^\star
= \Phi_\DD(f^\star) \cap g^\star \in \QQ_{\DD_1}$, where $f \in \FF$ and $g \in
\SS$, such that $\supp{f^\star}$ defines a $(k,d-k)$-colorful choice that
ray-embraces $\bb'$.

\begin{corollary}\label{stm:halfhalf_binary}
  Let $k \in [d-1]$, be a number and let $e, e' \in \QQ_{\DD_1}$ be two
  edges with $e = \Phi_\DD(f) \cap g$ and $e' = \Phi_\DD(f') \cap g'$, where $f,
  f' \in \FF$ and $g,g' \in \SS$. If $|\ind{C_1} \cap \supp{f}| < k$ and
  $|\ind{C_1} \cap \supp{f'}| > k$, then there exists an edge $e^\star =
  \Phi_\DD(f^\star) \cap
  g^\star \subset \conv(e,e')$, $e^\star \in \QQ_{\DD_1}$, such that
  $\supp{f^\star}$ defines a $(k, d-k)$-colorful choice of $C_1$ and $C_2$ that
  ray-embraces $\bb'$, where $f^\star \in \FF$ and $g^\star \in \SS$.
\end{corollary}
\begin{proof}
By Lemma~\ref{lem:edges}, the supports of the faces in $\FF$ that corresponds to two
adjacent edges in $\QQ_{\DD_1}$ differ in at most one column. Since $|\ind{C_1} \cap
\supp{f}| < k$, $|\ind{C_1} \cap \supp{f'}| > k$, and since $Q_{\DD_1}$ is a
polytopal complex, there must be an edge $e^\star =
\Phi_\DD(f^\star) \cap g^\star \in \QQ_{\DD_1}$ between $e$ and $e'$ such that
$|\ind{C_1} \cap \supp{f^\star}| = k$. By Lemma~\ref{stm:face_dim}, $f^\star$ is a
vertex and hence $|\supp{f^\star}| = d$. In particular, then $|\ind{C_2} \cap
\supp{f^\star}| = d-k$.
\end{proof}

The algorithm to find this $(k,d-k)$-colorful choice is now a straightforward
application of binary search. Initially we set $\mm_1 = \e_1$ and $\mm_2 =
\e_2$ and we maintain the invariant that the interval $[\mm_1,\mm_2]$ contains
an edge $e^\star = \Phi_\DD(f^\star) \cap g^\star \in \QQ_{\DD_1}$ such that
$\supp{f^\star}$ defines a $(k,d-k)$-colorful choice that ray-embraces
$\bb'$. The single optimal feasible
basis for $\e_1$ is $C_1$ and similarly, the single optimal feasible basis for
$\e_2$ is $C_2$. Then, Corollary~\ref{stm:halfhalf_binary} implies the invariant for the
initial interval. We repeatedly proceed as follows: set $\mm' = \frac{1}{2}
(\mm_1 + \mm_2)$ and solve the linear program $\LC_{\MM(\mm')}$. Let $\supp{f'}$
be the support of the maximum face $f' \in \FF$ that is optimal for
$\LC_{\MM(\mm')}$. First assume that $|\supp{f'}| = d$, i.e., assume that $f'$ is
a vertex of $\PC$. If $|\ind{C_1} \cap \supp{f'}| = k$, we have found the desired
solution. If $|\ind{C_1} \cap \supp{f'}| < k$, we set $\mm_2 =
\mm'$ and otherwise, if $|\ind{C_1} \cap \supp{f'}| > k$, we set $\mm_1 = \mm'$. By
Corollary~\ref{stm:halfhalf_binary}, the invariant is maintained. Now, assume that
$|\supp{f'}| = d+1$, i.e., assume that $f'$ is an edge of $\PC$. Then, by
Lemma~\ref{stm:face_dim}, $\mm' = \Phi_\DD(f') \cap g$ is a vertex of $\QQ_{\DD_1}$
and since $\mm' \in \relint \DD_1$, it is incident to two edges $e_1, e_2 \in
\QQ_{\DD_1}$ with $e_1= \Phi_\DD(f_1)\cap g$ and $e_2 = \Phi_\DD(f_2)\cap g$, where
$f_1$ and
$f_2$ are the two incident vertices to the edge $f'$. We compute both supports
$\supp{f_1}$ and $\supp{f_2}$ by checking every $d$-subset of $\supp{f'}$ whether
it constitutes a basis. Then, we check whether one of the two supports is a
$(k,d-k)$-colorful choice. If not, then by Lemma~\ref{lem:edges}, either
both supports contain less than $k$ columns from $C_1$ or both contain more than
$k$ columns from $C_1$. In the first case, we set $\mm_2 = \mm'$ and in the
second case, we set $\mm_1 = \mm'$. Again, Corollary~\ref{stm:halfhalf_binary} guarantees
that the invariant is maintained.

Clearly, each update of the interval $[\mm_1,\mm_2]$ needs weakly polynomial
time since $\Oh{d}$ linear programs are solved. Furthermore, the
number of the steps needed before a solution is found is logarithmic in the
length of the shortest edge. The following lemma shows that the minimum length
of an edge in $\QQ_{\DD_1}$ is at least exponentially small in the length of the
\CCP instance.

\begin{lemma}\label{lem:halfhalf_edgelength}
Let $L$ be the length of the binary encoding of the \CCP instance
$(C'_1,\dots,C'_d,\bb')$ and
let $e=[\mm_1,\mm_2] \in \QQ_{\DD_1}$ be an edge. Then, $-\log \|\mm_2 - \mm_1\| =
\Om{\poly L }$.
\end{lemma}
\begin{proof}
We write $e$ as $\Phi_\DD(f) \cap g$ and the two incident vertices as
$\mm_1 = \Phi_\DD(f_1) \cap g_1$ and $\mm_2 = \Phi_\DD(f_2) \cap g_2$, where
$\lt\{f,f_1,f_2\rt\} \subseteq \FF$ and $\lt\{g,g_1,g_2\rt\} \subseteq \SS$.
We denote with $\up{\mm}_1 = \MM(\mm_1)$ and with
$\up{\mm}_1 =
\MM(\mm_1)$ the vertices in $\QQ$ whose central projections onto $\DD$ resulted
in $\mm_1$ and $\mm_2$, respectively. Since $e$ is an edge, $\up{\mm}_1\neq
\up{\mm}_2$ and hence there is a $j \in[d]$ with $\lt(\up{\mm}_1\rt)_j \neq
\lt(\up{\mm}_2\rt)_j$. By Lemma~\ref{stm:face_dim}, $f$ is a vertex of $\PC$ and
$\supp{f} \subseteq \supp{f_i}$ for $i=1,2$. Let $B$ denote the columns in
$A_\supp{f}$. Then, we can express $\up{\mm}_i$, $i=1,2$, as the unique solution
to the linear system $\LR_{B, f_i}$ extended by the constraints $\mm \in
\MM(g_i)$. Now, Lemma~\ref{lem:lpsol} guarantees that the logarithm of
$\lt(\up{\mm}_i\rt)_j$,
$i \in [2]$, is a polynomial in the size of the linear system and hence in $L$.
Since $(\mm_1)_j \neq
(\mm_2)_j$, we have $=-\log \|\mm_2 - \mm_1\| = \Om{\poly L}$, as claimed.
\end{proof}

The described binary-search algorithm needs therefore only polynomial time in
$L$ to compute a $(k,d-k)$-colorful choice $C'$ for $C'_1$ and $C'_2$. Since $L$
is polynomial in the length of the of the
original instance $(C_1,\dots,C_d,\bb)$, the running time is weakly polynomial
in the length of the original instance. Furthermore, we can obtain a
$(k,d-k)$-colorful choice $C$ for $C_1$ and $C_2$ by
replacing the perturbed points in $C'$ with the original points in $C_1 \cup
C_2$. Lemma~\ref{stm:p3} then guarantees that $C$ ray-embraces $\bb$.

\end{document}